\documentclass[11pt]{article}

\usepackage{amssymb}
\usepackage{amsmath}
\usepackage{amsthm}
\usepackage{geometry}
\usepackage[square]{natbib}
\usepackage{amsfonts}
\usepackage{url}

\usepackage[ruled]{algorithm2e}

\SetAlFnt{\small}
\SetAlCapFnt{\small}
\SetAlCapNameFnt{\small}
\SetAlCapHSkip{0pt}
\IncMargin{-\parindent}

\newtheorem{theorem}{Theorem}
\newtheorem*{theorem*}{Theorem}
\newtheorem{corollary}[theorem]{Corollary}

\newtheorem{lemma}[theorem]{Lemma}
\newtheorem{definition}[theorem]{Definition}
\newtheorem{proposition}[theorem]{Proposition}
\newtheorem{claim}[theorem]{Claim}
\newtheorem{remark}[theorem]{Remark}
\newtheorem{observation}[theorem]{Observation}
\newtheorem{example}[theorem]{Example}
\numberwithin{theorem}{section}
\renewenvironment{proof}{\noindent\bf{Proof.~}\rm}{\hfill$\blacksquare$\bigskip}

\newcommand{\Kalton}{K_s}
\newcommand{\Kalweak}{K_w}
\newcommand{\Intplus}{\mathrm{\mathbb{N}_{\ge 0}}}
\newcommand{\Replus}{\mathrm{\mathbb{R}_{\ge 0}}}
\newcommand{\PS}{\mathrm{PS}}
\newcommand{\NS}{\mathrm{NS}}

\begin{document}
	
\title{{Approximate Modularity Revisited}%
	\thanks{The conference version appeared in STOC'17. Part of this work was done at Microsoft Research, Herzliya. Part of the work of U.~Feige was done while visiting Princeton University. This work has received funding from the Israel Science Foundation (grant No.~1388/16), the European Research Council under the European Union's Seventh Framework Programme (FP7/2007-2013) / ERC grant agreement No.~337122, and the European Union's Horizon 2020 research and innovation programme under the Marie Sklodowska-Curie grant agreement No.~708935.}
}
	
\author{Uriel Feige\thanks{Weizmann Institute of Science, Rehovot, Israel, {\tt uriel.feige@weizmann.ac.il}.}, 
Michal Feldman\thanks{Tel-Aviv University, Tel-Aviv, and Microsoft Research, Herzliya, Israel, {\tt michal.feldman@cs.tau.ac.il}.}, 
Inbal Talgam-Cohen\thanks{Hebrew University, Jerusalem, Israel, {\tt inbaltalgam@gmail.com}.}}

\date{}

\maketitle

\begin{abstract}
Set functions with convenient properties (such as submodularity) appear in application areas of current interest, such as algorithmic game theory, and allow for improved optimization algorithms. It is natural to ask (e.g., in the context of data driven optimization) how robust such properties are, and whether small deviations from them can be tolerated. We consider two such questions in the important special case of linear set functions.

One question that we address is whether any set function that approximately satisfies the modularity equation (linear functions satisfy the modularity equation exactly) is close to a linear function. The answer to this is positive (in a precise formal sense) as shown by Kalton and Roberts [1983] (and further improved by Bondarenko, Prymak, and Radchenko [2013]). We revisit their proof idea that is based on expander graphs, and provide significantly stronger upper bounds by combining it with new techniques. Furthermore, we provide improved lower bounds for this problem.

Another question that we address is that of how to learn a linear function $h$ that is close to an approximately linear function $f$, while querying the value of $f$ on only a small number of sets. We present a deterministic algorithm that makes only linearly many (in the number of items) nonadaptive queries, by this improving over a previous algorithm of Chierichetti, Das, Dasgupta and Kumar [2015] that is randomized and makes more than a quadratic number of queries. Our learning algorithm is based on a Hadamard transform.
\end{abstract}



\section{Introduction}
\label{sec:intro}
A set function $f$ over a universe $U$ of $n$ items assigns a real value $f(S)$ to every subset $S \subseteq U$ (including the empty set $\emptyset$). Equivalently, it is a function whose domain is the Boolean $n$-dimensional cube $\{0,1\}^n$, where each coordinate corresponds to an item, and a vector in $\{0,1\}^n$ corresponds to the indicator vector of a set. Set functions appear in numerous applications, some of which are briefly mentioned in Section~\ref{sub:related-work}.
Though set functions are defined over domains of size $2^n$, one is often interested in optimizing over them in time polynomial in $n$. This offers several challenges, not least of which is the issue of representing $f$. An explicit representation of the truth table of $f$ is of exponential size, and hence other representations are sought.

Some classes of set functions have convenient structure that leads to a polynomial size representation, from which the value of every set can easily be computed. A prime example for this is the class of linear functions.
Formally, a set function $f$ is \emph{linear} if there exist constants $c_0, c_1, \ldots, c_n$ such that for every set $S$, $f(S) = c_0 + \sum_{i \in S} c_i$.  The constants $(c_0, c_1, \ldots, c_n)$ may serve as a polynomial size representation of $f$. More generally, for set functions that arise naturally, one typically assumes that there is a so called {\em value oracle}, such that for every set $S$ one can query the oracle on $S$ and receive $f(S)$ in reply (either in unit time, or in polynomial time, depending on the context). The value oracle serves as an abstraction of either having some explicit polynomial time representation (e.g., a Boolean circuit) from which the value of $f$ on any given input can be computed, or (in cases in which set functions model some physical reality) having a physical process of evaluating $f$ on the set $S$ (e.g., by making a measurement).

Optimizing over general set functions (e.g., finding the maximum, the minimum, maximizing subject to size constraints, etc.)  is a difficult  task, requiring exponentially many value queries if only a value oracle is given, and NP-hard if an explicit representation is given. However, for some special classes of set functions various optimization problems can be solved in polynomial time and with only polynomially many value queries. Notable nontrivial examples are minimization of submodular set functions~\cite{Sch00}\cite{IwataFF01}, and welfare maximization when the valuation function of each agent satisfies the {\em gross substitutes} property (see, e.g., \cite{Pae14}). For the class of linear set functions, many optimization problems of interest can be solved in polynomial time, often by trivial algorithms.

A major concern regarding the positive algorithmic results for some nice classes of set functions is their {\em stability}. Namely, if $f$ is not a member of that nice class, but rather is only {\em close} to being a member (under some natural notion of {\em closeness}), is optimizing over $f$ still easy? Can one obtain solutions that are ``close" to optimal? Or are the algorithmic results ``unstable" in the sense that a small divergence from the nice class leads to a dramatic deterioration in the performance of the associated algorithms?

A complicating factor is that there is more than one way of defining {\em closeness}. For example, when considering two functions, one may consider the variational distance between them, the mean square distance, the Hamming distance, and more. The situation becomes even more complicated when one wishes to define how close $f$ is to a given \emph{class $C$} of nice functions (rather than to a particular function). One natural definition is in terms of a distance to the function $g \in C$ closest to~$f$. But other definitions make sense as well, especially if the class $C$ is defined in terms of properties that functions in $C$ have. For example, the distance from being submodular can be measured also by the extent to which the submodularity condition $f(S) + f(T) \ge f(S\cup T) + f(S\cap T)$ might be violated by $f$ (as in~\cite{LLN06,KC10}), or even by the so called {\em supermodular degree}~\cite{FeigeIzsak}.

Following the lead of  Chierichetti, Das,  Dasgupta, and Kumar~[2015], the goal of this work is to study questions such as the above in a setting that is relatively simple, yet important, namely, that of linear set functions. As illustrated by the work of~\cite{CDDK}, even this relatively simple setting is challenging.

\subsection{Our results and techniques}

\paragraph{Kalton constants}
One question that we address is the relation between two natural notions of being close to a linear function. The first notion is that of being close point-wise in an additive sense. We say that $f$ is {\em $\Delta$-linear} if there is a linear set function $g$ such that
$|f(S) - g(S)| \le \Delta$ for every set $S$. Under this notion, the smaller $\Delta$ is, the closer we consider $f$ to being linear. The other notion of closeness concerns a different but equivalent definition of linear functions, namely, as those functions that satisfy the {\em modular equation} $f(S) + f(T) = f(S \cup T) + f(S \cap T)$ for every two sets $S$ and $T$. This form of defining linear functions is the key to generalizing linear functions to other classes of functions of interest, and specifically to submodular functions that satisfy  $f(S) + f(T) \ge f(S \cup T) + f(S \cap T)$, and are considered to be the discrete analog of convex functions. Formally, we say that $f$ is
\begin{itemize}
	\item  $\epsilon$-modular if  $|f(S) + f(T) - f(S \cup T) - f(S \cap T)| \le \epsilon$ for every two sets $S$ and $T$;
	\item \emph{weakly} $\epsilon$-modular if the inequality holds for every two \emph{disjoint} sets $S$ and $T$. 
\end{itemize}	
The smaller $\epsilon$ is, the closer we consider $f$ to being linear (or equivalently, to being modular).

It can easily be shown that every $\Delta$-linear function is $\epsilon$-modular for $\epsilon \le 4\Delta$. Establishing a reverse implication is more difficult. In~\cite{CDDK} it was shown that every $\epsilon$-modular function is $\Delta$-linear for $\Delta = O(\epsilon \log n)$. However, the authors of~\cite{CDDK} were not aware of the earlier work of ~\cite{KR}, which already showed that $\Delta \le O(\epsilon)$ (this work was brought to our attention by Assaf Naor). We shall use $\Kalton$ to denote the smallest constant such that every $\epsilon$-modular function is $\Kalton \epsilon$-linear, and refer to $\Kalton$ as the {\em strong Kalton constant}. The bound provided in~\cite{KR} was $\Kalton \le \frac{89}{2}$, and this was subsequently improved by Bondarenko, Prymak, and Radchenko~\cite{BPR} to $\Kalton \le 35.8$.

The approach initiated by Kalton and Roberts (and used  almost as a blackbox in~\cite{BPR}) makes essential use of expander graphs in deriving upper bounds on $\Kalton$. We revisit this approach, simplify it and add to it new ingredients. Technically, the advantage that we get by our new ingredients is that we can use bipartite graphs in which only small sets of vertices expand, whereas previous work needed to use bipartite graphs in which large sets expand. This allows us to use expanders with much better parameters, leading to substantially improved upper bounds on $\Kalton$.
For concreteness, we prove in this paper the following upper bound.

\begin{theorem}[Upper bound, strong Kalton]
\label{thm:upper15}
Every $\epsilon$-modular function is $\Delta$-linear for $\Delta < 12.65\epsilon$. Hence $\Kalton < 12.65.$
\end{theorem}

We remark that our technique for upper bounding $\Kalton$ adds a lot of versatility to the expander approach, which is not exploited to its limit in the current version of the paper: the upper bound that we report strikes a balance between simplicity of the proof and quality of the upper bound.
Directions for further improvements are mentioned in Section \ref{appx:UB-strong}.

Obtaining good lower bounds on $\Kalton$ is also not easy. Part of the difficulty is that even if one comes up with a function $f$ that is a candidate for a lower bound, verifying that it is $\epsilon$-modular involves checking roughly $2^{2n}$ approximate modularity equations (one equation for every pair $S$ and $T$ of sets), making a computer assisted search for good lower bounds impractical. For set functions over $4$ items, we could verify that the worst possible constant is $1/2$. Checking $\epsilon$-modularity is much easier for symmetric functions (for which the value of a set depends only on its size), but we show that for such functions $\Kalton = \frac{1}{2}$, and this is also the case for $\epsilon$-modular functions that are submodular (see Section~\ref{sub:warm-up}). The only lower bound on $\Kalton$ that we could find in previous work is $\Kalton \ge \frac{3}{4}$, implicit in \cite{Pawlik}. We consider a class of functions that enjoys many symmetries (we call such functions {\em $(k,M)$-symmetric}), and for some function in this class we provide the following lower bound.

\begin{theorem}[Lower bound, strong Kalton]
\label{thm:lower1}
There is an integer valued set function over 70 items that is 2-modular and tightly 2-linear. Hence $\Kalton \ge 1$.
\end{theorem}

In addition, we shall use $\Kalweak$ to denote the smallest constant such that every {\em weakly} $\epsilon$-modular function is $\Kalweak \epsilon$-linear, and refer to $\Kalweak$ as the {\em weak Kalton constant}. For the weak Kalton constant we prove the following theorem.

\begin{theorem}[Upper and lower bounds, weak Kalton]
\label{thm:UpperLowerWeak}
Every weakly $\epsilon$-modular function is $\Delta$-linear for $\Delta < 24\epsilon$, so $\Kalweak < 24$. Moreover, there is an integer valued set function over 20 items that is weakly 2-modular and 3-linear, so $\Kalweak \ge 3/2$.
\end{theorem}

\paragraph{Learning algorithms}
Another part of our work concerns the following setting. Suppose that one is given access to a value oracle for a function $f$ that is $\Delta$-linear. In such a setting, it is desirable to obtain an explicit representation for some linear function $h$ that is close to $f$, because using such an $h$ one can approximately solve optimization problems on $f$ by solving them exactly on $h$. The process of learning $h$ involves two complexity parameters, namely the number of queries made to the value oracle for $f$, and the computation time of the learning procedure, and its quality can be measured by the distance $d(h,f) = \max_{S}\{|h(S) - f(S)|\}$ of $h$ from $f$. It was shown in~\cite{CDDK} that for every learning algorithm that makes only polynomially many value queries, there will be cases in which $d(h,f) \ge \Omega(\Delta\sqrt{n}/\log n)$. This might sound discouraging but it need not be. One reason is that in many cases $\Delta < \log n/\sqrt{n}$. For example, this may be the case when $f$ itself is actually linear, but is modeled as $\Delta$-linear due to noise in measurements of values of $f$. By investing more in the measurements, $\Delta$ can be decreased. Another reason is that the lower bounds are for worst case $f$, and it is still of interest to design natural learning algorithms that might work well in practice. Indeed, \cite{CDDK} designed a randomized learning algorithm that makes $O(n^2 \log n)$ nonadaptive queries and learns $h$ within distance $O(\Delta \sqrt{n})$ from $f$. %
We improve upon this result in two respects, one being reducing the number of queries, and the other being removing the need for randomization.

\begin{theorem}[Learning algorithm]
\label{thm:learning}
There is a deterministic polynomial time learning algorithm that given value oracle access to a $\Delta$-linear function on $n$ items, makes $O(n)$ nonadaptive value queries and outputs a linear function~$h$, such that $h(S)$ is $O(\Delta(1 + \sqrt{\min\{|S|,n-|S|\}}))$-close to $f(S)$ for every set $S$.
\end{theorem}

The learning algorithm is based on the Hadamard basis.
A similar technique is applied in \cite{DY08} for a different purpose (as brought to our attention by Moni Naor) -- attacking the privacy of an $n$-sized binary database, and recovering all but $o(n)$ entries via a linear number of queries.

The Hadamard basis is an orthogonal basis of $\mathbb{R}^n$ consisting of vectors with $\pm1$ entries, which can be constructed for every $n$ that is a power of 2.
It has the following property: for every two linear functions that are $O(\Delta)$-close on the basis vectors (where each basis vector can be interpreted as a difference between two sets, one corresponding to its $+1$ indices, the other to its $-1$ indices), the linear functions are $O(\Delta(1 + \sqrt{\min\{|S|,n-|S|\}}))$-close on \emph{every} set $S$. (This property follows from the vectors having large norms, and thus a large normalization factor by which the distance of $\Delta$ is divided.)
Given $O(n)$ value queries to a $\Delta$-linear function $f$, an algorithm can learn the values of the linear function $g$ that is $\Delta$-close to $f$ for the $n$ Hadamard basis vectors, up to an additive error of $O(\Delta)$. This is enough information to construct a linear function $h$ that is $O(\Delta(1 + \sqrt{\min\{|S|,n-|S|\}}))$-close to $g(S)$, and thus to $f(S)$, for every set $S$.

\subsection{Additional Related Work}
\label{sub:related-work}

The stability of linearity and the connection between approximate modularity and approximate linearity, which we study on the discrete hypercube, have been extensively studied in mathematics on continuous domains (e.g., Banach spaces). An early example is \cite{Hye41}, and this result together with its extensions to other classes of functions are known as the Hyers-Ulam-Rassias theory~\cite{Jung11}.

Our work is related to the literature on \emph{data-driven optimization} (see, e.g., \cite{BT14,SV,HS16,BRS16}).
This literature studies scenarios in which one wishes to optimize some objective function, whose parameters are derived from real-world data and so can be only approximately evaluated. Such scenarios arise in applications like machine learning \cite{KSG08}, sublinear algorithms and property testing \cite{BLR93}, and algorithmic game theory \cite{LLN06}.
This motivates the study of optimization of functions that satisfy some properties only approximately (e.g., submodularity \cite{KC10}, gross substitutes \cite{TimInbalJan16}, or convexity \cite{BLNR15}).
Our work is related to this strand of works in that we are also interested in functions that satisfy some property (modularity in our case) approximately, but unlike the aforementioned works, we are interested in the characterization and learning of functions, not in their optimization.

Learning of (exactly) submodular functions was studied by \cite{BH11} and \cite{GHIM09}, in the PMAC model and general query model, respectively. We shall discuss learning of approximately modular functions in a value query model.

\paragraph{Organization}

In Section \ref{sec:prelim} we present preliminaries. 
Section \ref{sec:upper} shows improved upper bounds on the weak and strong Kalton constants, and
Section \ref{sec:lower} shows improved lower bounds on the weak and strong Kalton constants.
Section \ref{sec:learning} describes how to learn an approximately linear function from an approximately modular one.

\section{Preliminaries}
\label{sec:prelim}
\subsection{Approximate Modularity and Linearity}
\label{sub:basic-obs}

Let $U=\{1,\dots,n\}$ be a ground set of $n\ge 2$ \emph{items} (also called elements).
For every set $S\subseteq U$ of items, let $\bar{S}$ denote its complement $U \setminus S$.
A \emph{collection} $\mathcal{G}$ is a multiset of sets, and its complement $\bar{\mathcal{G}}$ is the collection of complements of the sets in $\mathcal{G}$. Given a set function $f:2^U\to\mathbb{R}$, the \emph{value} of set $S$ is $f(S)$.
Throughout we focus on additive closeness, and say that two values $x,y\in\mathbb{R}$ are $\Delta$-close (or equivalently, that $x$ is $\Delta$-close to $y$) if $|x-y|\le\Delta$.

A set function $f$ is \emph{$\epsilon$-modular} if for every two sets $S,T\subseteq U$,
\begin{equation}
|f(S) + f(T) - f(S\cup T) - f(S\cap T)| \le \epsilon.\label{eq:epsilon-mod}
\end{equation}
If $\epsilon = 0$ then $f$ is \emph{modular}.
A set function $f$ is \emph{weakly $\epsilon$-modular} if Condition \eqref{eq:epsilon-mod} holds for every two \emph{disjoint} sets $S,T\subseteq U$. The following proposition 
shows the relation:

\begin{proposition}
	\label{pro:weak-to-mod}
	Every weakly $\epsilon$-modular set function is a $2\epsilon$-modular set function.
\end{proposition}

\begin{proof}
	Let $f$ be a weakly $\epsilon$-modular set function, we show it must be $2\epsilon$-modular. For every two sets $S$ and $T$, since $f$ is weakly $\epsilon$-modular we have that $f(T)=f(S\cap T)+f(T\setminus S)-f(\emptyset)\pm \epsilon$, and that $f(S\cup T)\pm\epsilon=f(S)+f(T\setminus S)-f(\emptyset)$. Therefore, $f(S)+f(T) = f(S) + f(S\cap T)+f(T\setminus S)-f(\emptyset)\pm \epsilon = f(S\cup T)+f(S\cap T)\pm 2\epsilon$, completing the proof.
\end{proof}

\begin{observation}[Approximate modularity for $\ge2$ sets]
	\label{obs:iterative}
	Let $f$ be a weakly $\epsilon$-modular set function, and let $S\subseteq U$ be a set of items with a partition $(S_1,\dots,S_s)$ (i.e., the disjoint union $S_1\uplus\dots\uplus S_s$ is equal to $S$). By iterative applications of weak $\epsilon$-modularity,
	$$
	f(S)+(s-1)(f(\emptyset) - \epsilon) \le
	f(S_1)+\dots+f(S_s)\le f(S)+(s-1)(f(\emptyset) + \epsilon).
	$$
\end{observation}

A set function $f$ is \emph{linear} if there exist constants $c_0, c_1, \ldots, c_n$ such that for every set $S$, $f(S) = c_0 + \sum_{i \in S} c_i$.
A linear set function is \emph{additive} if $c_0 = 0$. The \emph{zero function} has $c_i=0$ for every $0\le i\le n$. A set function $f$ is \emph{$\Delta$-linear} if there exists a linear set function $g$ that is \emph{$\Delta$-close} to $f$, i.e., $f(S)$ and $g(S)$ are $\Delta$-close for every set $S$. We say that a set function $f$ is \emph{tightly} $\Delta$-linear if it is $\Delta$-linear and for every $\Delta'<\Delta$, $f$ is not $\Delta'$-linear. A \emph{closest linear function} to $f$ is a linear set function $g$ that is $\Delta$-close to $f$ where $f$ is tightly $\Delta$-linear.


\begin{proposition}[Modularity = linearity]
\label{pro:modular-is-linear}
A set function is modular if and only if it is linear.
\end{proposition}

\begin{proof}
	Let $f$ be a modular function, we show it must be linear: Let $c_0=f(\emptyset)$, and let $c_i=f(\{i\})-c_0$. For every set $S$, by modularity, $f(S)=f(S_1)+f(S_2)-c_0$ for any $S_1,S_2$ whose disjoint union is $S$. Applying this equality recursively gives $f(S)=\sum_{i\in S}c_i + |S|c_0 - |S-1|c_0 = \sum_{i\in S}c_i + c_0$, as required.
	
	For the other direction, let $f$ be a linear function such that  $f(S)=c_0+\sum_{i\in S}c_i$. For any two sets $S,T$, it holds that $f(S)+f(T)=2c_0+\sum_{i\in S}c_i+\sum_{i\in T}c_i$, and $f(S\cup T)+f(S\cap T)=2c_0+\sum_{i\in S\cup T}c_i+\sum_{i\in S\cap T}c_i$.	The desired equality follows by $\sum_{i\in S}c_i+\sum_{i\in T}c_i=\sum_{i\in S\cup T}c_i+\sum_{i\in S\cap T}c_i$.
\end{proof}

\begin{proposition}
\label{pro:linear-to-modular}
Every $\Delta$-linear set function is $4\Delta$-modular.
This result is tight, even for symmetric functions.
\end{proposition}

\begin{proof}
	Let $f$ be a $\Delta$-linear function, and let $g$ be a linear function such that $g(S)=c_0+\sum_{i \in S} c_i$ and $|f(S) - g(S)| \le \Delta$ for every set $S$.
	For every two sets $S$ and $T$,
	$f(S)+f(T)-f(S \cup T)-f(S \cap T) \leq (c_0+\sum_{i\in S}c_i + \Delta) + (c_0+\sum_{i\in T}c_i + \Delta) - (c_0+\sum_{i\in S \cup T}c_i - \Delta) - (c_0+\sum_{i\in S \cap T}c_i - \Delta) = 4\Delta$.
	Similarly,
	$f(S)+f(T)-f(S \cup T)-f(S \cap T) \geq (c_0+\sum_{i\in S}c_i - \Delta) + (c_0+\sum_{i\in T}c_i - \Delta) - (c_0+\sum_{i\in S \cup T}c_i + \Delta) - (c_0+\sum_{i\in S \cap T}c_i + \Delta) = - 4\Delta$.
	This establishes that $f$ is $4\Delta$-modular.
	
	The proposition is tight, even for symmetric functions:
	Consider the $1$-linear function on $4$ items in which sets of size $0$ and $4$ are worth $0$, sets of size 1 and 3 are worth $-1$, and sets of size 2 are worth $+1$. If $S$ and $T$ are two different sets of size 2 that intersect, the modularity equation is violated by 4.
\end{proof}

We shall often refer to set functions whose closest linear function is the zero function.

\begin{observation}[\cite{CDDK}]
	\label{obs:zero-func-wlog}
	For every $\epsilon$-modular (resp., weakly $\epsilon$-modular) set function $f$ that is tightly $\Delta$-linear, there is an $\epsilon$-modular (resp., weakly $\epsilon$-modular) set function $f'$ that is tightly $\Delta$-linear and whose closest linear function is the zero function.
	
	The function $f'$ can be defined as follows: if $g$ is a closest linear function to $f$ then $f'(S)=f(S)-g(S)$ for every set~$S$.
\end{observation}

\subsection{Kalton Constants}
\label{sec:kalton-prelim}

Let $\Kalweak$ denote the smallest constant such that every weakly $\epsilon$-modular set function is $\Kalweak\epsilon$-linear. Let $\Kalton$ denote the smallest constant such that every $\epsilon$-modular set function is $\Kalton\epsilon$-linear. 
Notice that $\Kalton \le \Kalweak$. We refer to $\Kalweak$ as the weak Kalton constant (the possibility that there is such a constant $\Kalweak$ independent of $n$ was advocated in the works of Nigel Kalton), and to $\Kalton$ as the strong Kalton constant. Formally:

\begin{definition}
	$\Kalton\in\Replus$ (resp., $\Kalweak\in\Replus$) is the strong (weak) \emph{Kalton constant} if:
	\begin{itemize}
	\item for every $n\in\Intplus,\epsilon\in\Replus$, every (weakly) $\epsilon$-modular set function over $[n]$ is $\Kalton\epsilon$-linear ($\Kalweak\epsilon$-linear); and
	\item for every $\kappa<\Kalton$ ($\kappa<\Kalweak$) and for every $\epsilon\in\Replus$, there exists a sufficiently large $n$ and a (weakly) $\epsilon$-modular set function $f$ over $[n]$ such that $f$ is not $\kappa\epsilon$-linear (it is sufficient that there exist a (weakly) $\epsilon$-modular set function that is tightly $\Kalton\epsilon$-linear ($\Kalweak\epsilon$-linear)).
	\end{itemize}
\end{definition}

The propositions in Section~\ref{sub:basic-obs} imply the following corollaries on Kalton constants.

\begin{corollary} [of Proposition~\ref{pro:weak-to-mod}]
	\label{cor:weak-to-strong}
	$\Kalweak/2 \le \Kalton \le \Kalweak$.
\end{corollary}

\begin{corollary}[of Observation \ref{obs:zero-func-wlog}]
\label{cor:kalton-zero-close}
	If $\Kalton$ is the strong Kalton constant for set functions whose closest linear function is the zero function, then $\Kalton$ is the strong Kalton constant (for general set functions).
Similarly, if $\Kalweak$ is the weak Kalton constant for set functions whose closest linear function is the zero function, then $\Kalweak$ is the weak Kalton constant (for general set functions).
\end{corollary}


\subsection{Kalton Constants for Special Cases}
\label{sub:warm-up}

The analysis of Kalton constants becomes much easier in the following special cases. 

\subsubsection{Kalton Constants for Symmetric Set Functions}

Here we give tight bounds on $\Kalton$ for symmetric set functions. A set function $f$ is \emph{symmetric} if for every two sets $S$ and $T$, $|S| = |T| \implies f(S) = f(T)$. A symmetric set function over $[n]$ can be represented as a function over integers $f:\{0,\dots,n\}\to\mathbb{R}$.
The following example shows that $\Kalton\ge 1/2$ for symmetric set functions, and Proposition \ref{pro:symm} shows this is tight.

\begin{example}
	\label{ex:symm-LB}
	For every $n\in\mathbb{N}_{\ge 0}$, consider the symmetric set function $f_n$ over $[n]$ where $f_n([n])=-\epsilon$ and $f_n(S)=0$ for every other set $S$.
	Let $g_\delta$ be a symmetric linear set function over $[n]$ such that $g_\delta(\{j\})=-\frac{\epsilon j}{n} + \delta$ for every $0\le j\le n$. Observe that $g_\delta$ is $\max\{\delta, \epsilon-\frac{\epsilon}{n} - \delta\}$-close to $f_n$. When $\delta=\frac{\epsilon}{2}-\frac{\epsilon}{2n}$, the distance is minimized and $g_\delta$ is a closest linear function to $f_n$. Thus $f_n$ is tightly $\delta$-linear.
	This shows that for every $\kappa <1/2$, there exists $n$ such that $f_n$ is tightly $\delta$-linear for $\delta = \frac{\epsilon}{2}-\frac{\epsilon}{2n}>\kappa\epsilon$.
\end{example}

\begin{proposition}
	\label{pro:symm}
	The strong Kalton constant for symmetric set functions is $\Kalton=\frac{1}{2}$.
\end{proposition}

\begin{proof}
	Let $f$ be a symmetric $\epsilon$-modular set function, we argue that $f$ must be $\frac{1}{2}\epsilon$-linear.
	By Observation~\ref{obs:zero-func-wlog-symm} we can assume without loss of generality that $f$'s closest linear function is the zero function.
	Let $M$ be the maximum absolute value of $f$, then there exist $k_1 < k < k_2$ such that either $f(k_1)=f(k_2)=-M$ and $f(k)=M$, or $f(k_1)=f(k_2)=M$ and $f(k)=-M$. Without loss of generality assume the former.
	Suppose that $k \geq n/2$. Using $\epsilon$-modularity we get that $2f(k) \leq f(k_2) + f(2k-k_2) + \epsilon \leq -M + M + \epsilon$. On the other hand, $2f(k) = 2M$. We get $2M \leq -M+M+\epsilon$, implying that $M \leq \epsilon / 2$, as desired. If $k < n/2$, an analogous argument is invoked using $k_1$.
\end{proof}

\begin{observation}
	\label{obs:zero-func-wlog-symm}
	For every symmetric $\epsilon$-modular set function $f$ that is tightly $\Delta$-linear, there is a \emph{symmetric} $\epsilon$-modular set function $f'$ that is tightly $\Delta$-linear and whose closest linear function is the zero function.
\end{observation}

\begin{proof}
	Let $g$ be a linear set function $\Delta$-close to $f$. We show a symmetric set function $g'$ that is $\Delta$-close to $f$: for every $k\in[n]$, let $g'(k)$ be the average value of $k$ items according to $g$. The proof follows as in Observation \ref{obs:zero-func-wlog}.
\end{proof}

\subsubsection{Kalton Constants for Submodular Set Functions}

We now give tight bounds on $\Kalton$ for submodular set functions. A set function $f$ is \emph{submodular} if for every two sets $S,T$ it holds that $f(S)+f(T)\ge f(S\cup T)+f(S\cap T)$. Example \ref{ex:symm-LB} shows that $\Kalton\ge 1/2$ not only for symmetric set functions but also for submodular ones.

\begin{proposition}
	\label{pro:submod}
	The strong Kalton constant for submodular set functions is $\Kalton=\frac{1}{2}$.
\end{proposition}

\begin{proof}
	Without loss of generality we normalize $f$ such that $f(\emptyset)=0$.
	Since $f$ is submodular and normalized, it belongs to the class of XOS functions, and so there exists an additive set function $g$ such that for every set $S$, $f(S)\ge g(S)$, and for the set $S=U$, $f(U)=g(U)$  (see, e.g., \cite{Fei09}). 
	It remains to show that for every $S$, $f(S) \le g(S)+\epsilon$. Assume for contradiction that $f(S)>g(S) +\epsilon$. Then $f(S)+f(\overline{S})> g(S)+\epsilon+g(\overline{S}) = g(U)+\epsilon = f(U)+ \epsilon$, where we use that $g(S)+g(\overline{S}) = g(U)$ by additivity. We have shown a contradiction to $\epsilon$-modularity, completing the proof.
\end{proof} 

\subsection{Results from the Literature}
\label{sub:related}

\subsubsection{Upper and Lower Bounds}

\cite{KR} prove that $\Kalweak \le \frac{89}{2}$. This upper bound was subsequently improved to $\Kalweak \le 38.8$ by~\cite{BPR}, who also show that $\Kalton \le 35.8$. Let us provide more details on how the known upper bounds on $\Kalweak$ are achieved.

\begin{definition}[Expander]
	\label{def:expander}
	For $k\in \mathbb{N}_{\ge 0}$ and $\alpha,r,\theta\in \Replus$ such that $\alpha,\theta < 1$ and $r > 2$, we say that a bipartite graph $G_k(V,W;E)$ is an $(\alpha, r,\theta)$-expander if $|V| = 2k$, $|W| = 2\theta k$, $|E| = 2kr$, and every set $S \subset V$ of at most $2k\alpha$ vertices has at least $|S|$ neighbors in $W$ (and hence a perfect matching into $W$).
\end{definition}

We say that $(\alpha, r,\theta)$-expanders exist if there is some $k'\in\mathbb{N}_{\ge 0}$ such that for every integer multiple $k$ of $k'$, there exists an $(\alpha,r,\theta)$-expander $G_k$.

The following theorem (rephrased from~\cite{KR}) is the key to the known upper bounds on $\Kalweak$.

\begin{theorem}[\cite{KR}]
	\label{thm:KR}
	Suppose that for fixed $r$ and $\theta$ and all sufficiently large $k$ there are $(\frac{1}{2},r,\theta)$-expanders $\{G_k\}$. Then an upper bound on the weak Kalton constant is:
	$$\Kalweak \le \frac{7+ 4r - 2\theta}{2(1 - \theta)}.$$
\end{theorem}

\cite{pippenger} shows that $(\frac{1}{2},r,\theta)$-expanders exist with $r=6$ and $\theta = \frac{2}{3}$, if $k$ is sufficiently large ($k\ge 3$). This together with Theorem~\ref{thm:KR} implies that $\Kalweak \le \frac{89}{2}$. In~\cite{BPR} it is shown that $r$ can be reduced to $5.05$, thus leading to the improved bound of $\Kalweak \le 38.8$.

As for lower bounds, \cite{Pawlik} gives a high-level sketch of a construction that shows the following for the weak Kalton constant. His construction implies also a lower bound for the strong Kalton constant, as follows (a detailed proof appears for completeness in Appendix \ref{appx:prelim}).

\begin{theorem}[Lower bound \cite{Pawlik}]
	\label{thm:pawlik}
	Lower bounds on the Kalton constants are $\Kalweak\ge \frac{3}{2}-\Theta(1/n)$ and $\Kalton\ge \frac{3}{4}-\Theta(1/n)$.
\end{theorem}

\subsubsection{A Characterization of Chierichetti et al.~\cite{CDDK}}

For a set function $f$, let $M=\max_{S}\{|f(S)|\}$ be the maximum absolute value of $f$ (also called $f$'s \emph{extreme} value). We say that a set $S$ has value $M$ if $f(S)=M$, and value $-M$ if $f(S)=-M$. Given a distribution $(p_1,\dots,p_\kappa)$ over sets $S_1,\dots,S_\kappa$, the \emph{marginal probability} of item $i$ according to this distribution is the probability that $i$ appears in a set randomly selected according to the distribution, i.e., $\sum_{j\mid i\in S_j}p_j$.

\begin{lemma}[\cite{CDDK}]
	\label{lem:charac}
	The closest linear function to a set function $f$ is the zero function if and only if there exist probability distributions $P^+$ and $P^-$ with rational probabilities over sets with value $M$ and sets with value $-M$, respectively, such that for every item $i$, the marginal probabilities of $i$ according to $P^+$ and $P^-$ are the same.
\end{lemma}

Lemma \ref{lem:charac} appears as Lemma 10 in \cite{CDDK}. The rationality of the probability distributions follows since they are obtained as solutions to a linear program.

\begin{definition}[Positive and Negative Supports]
	\label{def:PS-NS}
	For a set function $f$ whose closest linear set function is the zero function, let $P^+,P^-$ be the distributions guaranteed by Lemma \ref{lem:charac}. Then the \emph{positive support} $\PS=(P_1,\dots,P_\kappa)$ of $f$ is the support of $P^+$ (sets assigned a positive probability by $P^+$), and the {\em negative support} $\NS=(N_1,\dots,N_\nu)$ of $f$ is the support of $P^-$.
\end{definition}

We emphasize that all sets in $\PS$ have value $M$ and all sets in $\NS$ have value $-M$. 

\section{Improved Upper Bounds}
\label{sec:upper}
\subsection{Our Approach}
\label{sub:our-approach}

Our main goal in this section is to provide improved upper bounds for the strong Kalton constant~$\Kalton$, which was previously known to be $\le 35.8$ \cite{BPR}. Along the way we also provide an improved upper bound for the weak Kalton constant~$\Kalweak$, which was previously known to be $\le 38.8$ \cite{BPR}.

Let us first review the known upper bounds on $\Kalweak$ and how they are derived. The basic approach of \cite{KR} is outlined in Theorem~\ref{thm:KR}. There the value of $\Kalweak$ is related to parameters of $(\frac{1}{2},r,\theta)$-expanders. Rearranging the bound from Theorem~\ref{thm:KR}, denoting the denominator by $D$ and the numerator by $N_r + N_1$ where $N_r$ depends on $r$ and $N_1$ does not, it is shown that:
$$
\Kalweak \le \frac{N_r + N_1}{D},
$$
where $D = 1 - \theta$, $N_r = 2r$ and $N_1 = \frac{7}{2} - \theta$. \cite{KR} use a previously known expander construction of \cite{pippenger} to get an upper bound of 44.5.
The improved upper bound of 38.8 of~\cite{BPR} comes from constructions of expanders with a smaller value of~$r$.
The value of $r$ cannot be substantially reduced further (without changing $\theta$), and so the approach of constructing better expanders is unlikely to significantly further reduce $\Kalweak$.

In Section \ref{sub:UB-weak} we improve upon the upper bound of 38.8 on $\Kalweak$ by reducing $N_1$ from  $\frac{7}{2} - \theta$ to $-\frac{1}{2} - \theta$,
giving an upper bound of 26.8 using the expander construction of \cite{BPR}. The key to this improvement is extracting the main idea from the proof of~\cite{KR}, cleaning away redundancies
and using instead Lemma~\ref{lem:canonical} (which establishes the existence of complementary collections with values approximately equal to the function's extreme value or its negation). 
It seems that the value of $N_1$ cannot be substantially reduced further within this framework of ideas,
and hence new ideas appear to be needed if one wishes to obtain significant improvements in the upper bound on $\Kalweak$. In Appendix~\ref{appx:UB-weak} we upper-bound $\Kalweak$ using the minimum between two expressions rather than a bound of the form $(N_r+N_1)/D$, and show that $\Kalweak\le 23.82$.

In Sections \ref{sub:UB-strong} and \ref{appx:UB-strong} we consider $\Kalton$, the strong Kalton constant. This places additional restrictions on $f$ ($\epsilon$-modularity instead of only weak $\epsilon$-modularity). Indeed, these additional restrictions were used in~\cite{BPR} to reduce $N_1$ from $\frac{7}{2} - \theta$ to $\frac{5}{2} - \theta$, achieving an upper bound of $35.8$ in the case of $\Kalton$.

Our approach for improving $\Kalton$ will make more extensive use of $\epsilon$-modularity.
Rather than considering sets from a collection with approximately the extreme value $M$, we shall consider \emph{intersections} of these sets. Using $\epsilon$-modularity we shall be able to show that the function values of intersections are also close to $M$. In fact, using some averaging arguments we shall obtain even stronger estimates on how close these values are to $M$ (a point that is relevant to controlling the value of $N_1$). Thereafter, we will no longer be restricted to using $(\alpha,r,\theta)$-expanders with $\alpha = \frac{1}{2}$. We will be able to use $\alpha = \frac{1}{4}$ instead (for intersections of two sets), or even $\alpha = \frac{1}{8}$ (for intersections of three sets), and so on. The advantage of reducing $\alpha$ is that for smaller values of $\alpha$, expanders with lower values of $r$ and~$\theta$ exist, leading to better upper bounds on $\Kalton$. However, we cannot reduce $\alpha$ to arbitrarily small values because each reduction of $\alpha$ by a factor of two is accompanied by an increase in $N_1$, and one needs to balance between these two factors.

\subsection{Upper Bounds Preliminaries}

We present definitions and preliminary results used to establish our upper bounds. 

\subsubsection{Expanders}
\label{sub:expanders}

As discussed when describing our approach (Section \ref{sub:our-approach}), we utilize the existence of expanders with a range of parameters. Our existence argument (see Appendix \ref{appx:expanders}) uses the probabilistic method as in~\cite{pippenger}, and so results in expanders that are biregular (all vertices on the same side of the bipartite graph have the same degree). More complicated expander constructions that are not biregular may achieve even better parameters, as in \cite{BPR}.

\begin{lemma}
	\label{lem:expander}
	The following families of (biregular) expanders exist:
	\begin{enumerate}
		\item $(\frac{1}{2},5,\frac{5}{7})$-expanders.
		\item $(\frac{3}{10},4,\frac{4}{7})$-expanders.
		\item $(\frac{1}{4},5,\frac{1}{2})$-expanders.
		\item $(\frac{1}{16},4,\frac{4}{15})$-expanders.
		\item $(\frac{1}{64},3,\frac{3}{11})$-expanders.
		\item $(\frac{1}{256},3,\frac{3}{19})$-expanders.
	\end{enumerate} 
\end{lemma}

\subsubsection{Weakly $\epsilon$-modular set functions}

Throughout this section, let $f$ be a weakly $\epsilon$-modular set function whose closest linear set function is the zero function.

\begin{observation}
	\label{obs:empty-plus-full}
	$-\epsilon\le f(\emptyset)+f(U)\le \epsilon$.
\end{observation}

\begin{proof}
	Denote $\delta=f(\emptyset)+f(U)$.
	Let $M$ denote the maximum absolute value of $f$. If $M=0$ the claim follows trivially, otherwise by Lemma \ref{lem:charac} there exist sets $P,N$ with values $M,-M$, respectively. Consider the values of $\bar{P},\bar{N}$. By weak $\epsilon$-modularity and the definition of $M$, we have $-M\le f(\bar{P})\le -f(P)+\delta+\epsilon=-M+\delta+\epsilon$, and $M+\delta-\epsilon=-f(N)+\delta-\epsilon\le f(\bar{N})\le M$. We conclude that $0\le \delta+\epsilon$ and $\delta-\epsilon\le 0$, completing the proof.
\end{proof}

The following is a direct corollary of Observation \ref{obs:empty-plus-full} and weak $\epsilon$-modularity:

\begin{corollary}[Value of complement set]
	\label{cor:comp-val}
	Let $\delta=f(\emptyset)+f(U)$. Then for every set $S$,
$-f(S)-\epsilon+\delta \le f(\bar{S})\le -f(S)+\epsilon+\delta$.
\end{corollary}

Let $M$ be the maximum absolute value of $f$. 

\begin{definition}
A set $S$ has {\em deficit} $d\ge 0$ if $f(S) = M - d$, and has {\em surplus} $s\ge 0$ if $f(S) = -M + s$. 
A collection has \emph{average} deficit $d$ (resp.,~surplus $s$) if the expected deficit (resp.~surplus) of its sets with respect to the uniform distribution is $d$ (resp.~$s$). 
\end{definition}

The next observation follows directly from Corollary \ref{cor:comp-val}.

\begin{observation}[Average deficit/surplus of complement collection]
	\label{cor:complement-deficit}
	Let $\delta=f(\emptyset)+f(U)$.
	Let $\mathcal{G}$ be a collection with average deficit $d$ and surplus $s$. Then the average surplus of its complement $\bar{\mathcal{G}}$ is at most $d+\epsilon+\delta$, and the average deficit of $\bar{\mathcal{G}}$ is at most $s+\epsilon-\delta$.
\end{observation}

\begin{definition}
	An item is \emph{$\alpha$-frequent} in a collection if it appears in exactly an $\alpha$-fraction of the sets. A collection has $\alpha$-frequent items if every item is $\alpha$-frequent in it.
\end{definition}

\begin{observation}
	\label{obs:freq-in-complement}
	The complement of a collection with $\alpha$-frequent items has $(1-\alpha)$-frequent items.
\end{observation}

\begin{lemma}[Complementary collections]
	\label{lem:canonical}
	There exists $k'\in \mathbb{N}_{\ge 0}$ such that for every $k$ which is an integer multiple of $k'$, $f$ has a collection $\PS^*$
	of $2k$ sets (the same set might appear multiple times in the collection and we treat these appearances as distinct)
	with $1/2$-frequent items and average deficit $d$, whose complement collection $\NS^*$ 
	has $1/2$-frequent items and average surplus $s$, and $d+s\le \epsilon$.
\end{lemma}

\begin{proof}
	Consider the positive and negative supports $\PS,\NS$ of $f$, where $\PS=\{P_1,\dots,P_\kappa\}$ and $\NS=\{N_1,\dots,N_{\nu}\}$, as defined in Definition \ref{def:PS-NS}. By Lemma \ref{lem:charac}, there exist distributions $P^+=(p^+_1,\dots,p^+_\kappa)$ and $P^-=(p^-_1,\dots,p^-_\nu)$ with rational probabilities over $\PS$ and $\NS$, whose marginals are equal for all items. Consider the complement collections $\overline{\PS}$ and $\overline{\NS}$; let $d'$ be the average deficit of $\overline{\NS}$ and let $s'$ be the average surplus of $\overline{\PS}$. Since the average deficit of $\PS$ and the average surplus of $\NS$ are 0, then by Corollary \ref{cor:complement-deficit}, $d'\le \epsilon-\delta$ and $s'\le \epsilon+\delta$ and so $d'+s'\le 2\epsilon$.
	
	Towards constructing the collections $\PS^*$ and $\NS^*$, consider the collections $\PS\cup\overline{\NS}$ and $\overline{\PS}\cup\NS$ of $\kappa+\nu$ sets each. 
	We define a distribution $Q$ over $\kappa+\nu$ sets, which can be associated with both $\PS\cup\overline{\NS}$ and $\overline{\PS}\cup\NS$, to be
	$Q=(\frac{1}{2}p^+_1,\dots,\frac{1}{2}p^+_\kappa, \frac{1}{2}p^-_1,\dots,\frac{1}{2}p^-_\nu)$. This distribution has the following properties:
	\begin{itemize}
		\item First, if sets are randomly drawn from $\PS\cup\overline{\NS}$ according to this distribution, the probability of selecting a set with value $M$ is at least $1/2$, since the total weight on sets in $\PS$ is exactly $1/2$. Similarly, the probability of selecting a set from $\overline{\PS}\cup\NS$ with value $-M$ when sampling according to $Q$ is at least $1/2$.
		\item Second, for $\PS\cup\overline{\NS}$ and for every item $i$, the probability of selecting a set with item $i$ (i.e., the marginal of $i$) is exactly $1/2$, since it is equal to $\sum_{j\mid i\in P_j}\frac{1}{2}p^+_j + \sum_{j\mid i\in \bar{N}_j}\frac{1}{2}p^-_j = \sum_{j\mid i\in P_j}\frac{1}{2}p^+_j + \frac{1}{2}- \sum_{j\mid i\notin \bar{N}_j}\frac{1}{2}p^-_j$, and from the equality of the marginals of $P^+,P^-$ we have that	
		$\sum_{j\mid i\in P_j} p^+_j=\sum_{j\mid i\in N_j} p^-_j=\sum_{j\mid i\notin \bar{N}_j} p^-_j$. The same holds for item marginals when the distribution is taken over $\overline{\PS}\cup\NS$.
		
		\item Third, the probabilities of the distribution $Q$ are all rational and strictly positive.
	\end{itemize}
	
	By the third property, we can duplicate sets in $\PS\cup\overline{\NS}$ and in $\overline{\PS}\cup\NS$ to construct the collections $\PS^*$ and $\NS^*$, such that sampling a set uniformly at random from $\PS^*$ is equivalent to sampling a set according to $Q$ from $\PS\cup\overline{\NS}$, and similarly for $\NS^*$ and $\overline{\PS}\cup\NS$.
	By the first property, the average deficit of $\PS^*$ is $d\le d'/2$ and the average surplus of $\NS^*$ is $s\le s'/2$, and $d+s\le \epsilon$.
	By the second property, every item must appear in exactly half the sets in $\PS^*$ and half the sets in $\NS^*$, meaning that the number of sets in $\PS^*$ and $\NS^*$ is even, and we can denote it by $2k'$ for some integer $k'>0$. We have thus shown the existence of a collection $\PS^*$ and its complement $\NS^*$ with $k'$ sets each whose average surplus and deficit guarantee $d+s\le \epsilon$. Observe that existence of such collections of size $k$ holds for every integer multiple $k=ck'$, since taking $c$ copies of $\PS^*$ and $c$ copies of $\NS^*$ satisfies all the conditions of the lemma.
\end{proof}

\subsection{Main Lemmas for Upper Bounds}

We present the two key lemmas used to establish our upper bounds.
We begin with Lemma \ref{lem:KR}, which is a simplified version of Lemma~3.1 of~\cite{KR}.

\begin{lemma}[Using expanders for set recombination]
\label{lem:KR}
Let $k\in \mathbb{N}_{\ge 0}$ and $\alpha,r,\theta\in \Replus$ be such that there exists an $(\alpha,r,\theta)$-expander $G_k$. Consider a collection $\mathcal{G}$ of $2k$ sets with $\alpha$-frequent items (referred to as the \emph{source} sets). Then $\mathcal{G}$ has a refined partition into a total of $2kr$ subsets (referred to as the \emph{intermediate} subsets), which can be recombined by \emph{disjoint} unions into a collection of $2k\theta$ sets with  $\alpha/\theta$-frequent items (referred to as the \emph{target} sets).
\end{lemma}

\begin{proof}
Let $G_k=G_k(V,W;E)$.
Align the $2k$ source sets with the $2k$ vertices of $V$ in the $(\alpha,r,\theta)$-expander $G_k$.
Because every item appears in $2k\alpha$ sets,
then for every item~$i$ there are $2k\alpha$ vertices in $V$ corresponding to $i$ (i.e., aligned with the source sets that contain $i$). By the expansion property of $G_k$, for every item~$i$ there exists a perfect matching $M_i$ between the vertices in $V$ that correspond to $i$, and some $2k\alpha$ vertices in $W$.

We now use these matchings to label the edges: For every item $i$, add $i$ to the labels of the matched edges in $M_i$. Every edge in $E$ is now labeled by a subset of items (some labels may be the empty set). Let these labels be the intermediate subsets. Their total number is $|E| = 2kr$, as desired. For every vertex $v\in V$, the items of the source set $S_v$ corresponding to $v$ are partitioned among the edges leaving $v$, and hence the intermediate subsets indeed reflect a refined partitioning of the target sets.

Observe that the edges entering a vertex $w \in W$ are labeled by disjoint intermediate subsets (since for every item $i$, the edges labeled by subsets containing $i$ form a matching). Let the set $S_w$ corresponding to $w$ be the disjoint union of the subsets labeling the edges adjacent to $w$. The sets corresponding to the vertices in $W$ can thus be the target sets.

Notice that by construction, every item $i$ appears in the same number of source and target sets, so if the source sets have $\alpha$-frequent items, the target sets have $\alpha/\theta$-frequent items, completing the proof.
\end{proof}

Let $f$ be a weakly $\epsilon$-modular set function whose closest linear set function is the zero function. Let $M$ be the absolute highest value of $f$. The following lemma upper-bounds $M$; a more nuanced version appears as Lemma \ref{lem:UB-main-lemma-full} in Appendix \ref{appx:UB-weak}.

\begin{lemma}
	\label{lem:UB-main-lemma}
	Let $k\in \mathbb{N}_{\ge 0}$ and $\alpha,r,\theta\in \Replus$ be such that there exists an $(\alpha,r,\theta)$-expander $G_k$.
	Let $\mathcal{G}$ and $\mathcal{G}'$ be collections of $2k$ sets each, both with $\alpha$-frequent items, such that the average deficit of $\mathcal{G}$ is at most $d$ and the average surplus of $\mathcal{G}'$ is at most $s$. Then
	\begin{equation*}
	M\le \frac{\frac{1}{2}(d+s) + 2\epsilon(r - 1)}{1-\theta} + \epsilon.
	\end{equation*}
\end{lemma}

\begin{proof}
	We first apply Lemma~\ref{lem:KR} to partition and disjointly recombine the sets of collection $\mathcal{G}$ using the expander $G_k=G_k(V,W;E)$. We use the following notation: For a vertex $v\in V$ (resp., $w\in W$) let $S_v$ (resp., $S_w$) be the source (resp., target) set corresponding to $v$ (resp., $w$). Denote the neighboring vertices of a vertex $v$ by $N(v)$ and its degree by $\deg(v)$. Let $S_{v,w}$ be the intermediate subset that labels (corresponds to) edge $(v,w)\in E$.
	
	By Lemma \ref{lem:KR}, for every $v\in V$, the intermediate subsets labeling the edges adjacent to $v$ are disjoint, and the same holds for every $w\in W$. We can thus apply Observation \ref{obs:iterative} to get
	\begin{eqnarray*}
		&\sum_{w \in N(v)} f(S_{v,w}) \ge f(S_v) + (\deg(v) - 1)(f(\emptyset) - \epsilon)&\text{~~~}\forall v\in V.\\
		&\sum_{v \in N(w)} f(S_{v,w}) \le f(S_w) + (\deg(w) - 1)(f(\emptyset) + \epsilon)&\text{~~~}\forall w\in W.
	\end{eqnarray*}
	Denote the maximum absolute value of $f$ by $M$. Since collection $\mathcal{G}$ has average deficit of at most $d$, by summing over vertices $v\in V$ (where $|V|=2k$) we get $\sum_{v\in V}f(S_v)\ge 2k(M-d)$. Since the target sets of $\mathcal{G}$ have value at most $M$, by summing over vertices $w\in W$ (where $|W|=2\theta k$) we get $\sum_{w\in W}f(S_w)\le 2\theta kM$. Clearly in the bipartite graph $G_k$, $\sum_{v\in V}\deg(v)=\sum_{w\in W}\deg(w)$, and by the parameters of $G_k$ both are equal to $2kr$. Therefore, summing over $v \in V$ and $w\in W$ we get
	\begin{eqnarray*}
	2kM - 2kd + (2kr - 2k)(f(\emptyset) - \epsilon)
	\le 
	\sum_{(v,w)\in E} f(S_{v,w}) \\
	\le 2\theta kM + (2kr - 2\theta k)(f(\emptyset) + \epsilon).
	\end{eqnarray*}
	Dividing the resulting inequality by $2k$ and rearranging gives
	\begin{eqnarray}
	(1 - \theta)M
	&\le&
	d + (r - 1)(\epsilon - f(\emptyset)) + (r - \theta)(\epsilon + f(\emptyset))\nonumber\\
	&=& d + 2\epsilon(r - 1) + (1-\theta)(\epsilon+f(\emptyset)).
    \label{eq:positiveM}
	\end{eqnarray}
	
	Similarly, using that the average deficit of collection $\mathcal{G}'$ is at most $s$ and the value of its target sets is at most $M$, $\sum_{v\in V}f(S_v)\le 2k(-M+s)$, and $\sum_{w\in W}f(S_w)\ge -2\theta kM$. Therefore
	\begin{eqnarray*}
	- 2kM + 2ks + (2kr - 2k)(f(\emptyset) + \epsilon)
	\ge \sum_{(v,w)\in E} f(S_{v,w})\\
	\ge -2\theta kM + (2kr - 2\theta k)(f(\emptyset) - \epsilon).
	\end{eqnarray*}
	Dividing the resulting inequality by $2k$ and rearranging gives
	\begin{eqnarray}
	(1 - \theta)M
	&\le&
	s + (r - 1)(\epsilon + f(\emptyset)) + (r - \theta)(\epsilon - f(\emptyset))\nonumber\\
	&=& s + 2\epsilon(r - 1) + (1-\theta)(\epsilon-f(\emptyset)).
	\label{eq:negativeM}
	\end{eqnarray}
	
	Rearranging Inequalities \eqref{eq:positiveM} and \eqref{eq:negativeM} as well as averaging the resulting inequalities implies the theorem.
\end{proof}

\subsection{Upper-Bounding the Weak Kalton Constant $\Kalweak$}
\label{sub:UB-weak}

To upper bound $\Kalweak$ in this section and in Appendix \ref{appx:UB-weak}, we may focus without loss of generality on a weakly 1-modular set function~$f$, whose closest linear set function is the zero function (Corollary \ref{cor:kalton-zero-close}).
We show the following upper bound:

\begin{lemma}[Upper bound on $\Kalweak$]
	\label{lem:K-first}
	Suppose that for fixed $r,\theta\in \Replus$
	there exist $(\frac{1}{2},r,\theta)$-expanders.
	Then the weak Kalton constant satisfies:
	$\Kalweak \le
	\frac{ 2r - \frac{1}{2} - \theta} {1-\theta}.$
\end{lemma}

\begin{proof}
	By Lemma~\ref{lem:canonical}, there exists $k'$ such that for every $k$ that is a member of the arithmetic progression $k', 2k', 3k', \ldots$, the function $f$ has collections $\PS^*$ and $\NS^*$ of $2k$ sets with $\frac{1}{2}$-frequent items, and whose average deficit $d$ and average surplus $s$, respectively, satisfy $d+s\leq1$.
The assumption that $(\frac{1}{2},r,\theta)$-expanders exist implies that there is another arithmetic progression $k", 2k", 3k", \ldots$ such that for every $k$ in this sequence an expander $G_k$ with the above parameters exist. As the two arithmetic progressions must meet at $k'k"$, there is a common $k$ in both progressions.
	The upper bound follows from applying Lemma \ref{lem:UB-main-lemma} to collections $\PS^*,\NS^*$ to get
	$\Kalweak \le \frac{\frac{1}{2}+2r-2}{1-\theta} + 1$.
\end{proof}

\begin{theorem}
	The weak Kalton constant satisfies $\Kalweak \le 26.8$.
\end{theorem}

\begin{proof}
	In Lemma \ref{lem:K-first}, we can get $K_w\le 26.8$ if we substitute $r=5.05$ and $\theta=\frac{2}{3}$ by using the expanders from~\cite{BPR}.%
	\footnote{Using the expanders from~\cite{pippenger} (with $r=6$ instead of $r=5.05$) would result in $K_w\le 32.5$.}
\end{proof}

Using additional ideas, we further improve our upper bound on the weak Kalton constant by providing a stronger version of Lemma \ref{lem:UB-main-lemma}.
Using the new ideas we show that the weak Kalton constant satisfies $\Kalweak \le 23.811$, thus proving the upper bound of Theorem~\ref{thm:UpperLowerWeak}.
The full proof is deferred to Appendix~\ref{appx:UB-weak}.

\subsection{Upper-Bounding the Strong Kalton Constant $\Kalton$}
\label{sub:UB-strong}

As in the previous section, we focus here without loss of generality on a 1-modular set function~$f$ whose closest linear set function is the zero function, and on its collections $\PS^*$ and $\NS^*$ as defined in Lemma \ref{lem:canonical}. Recall that $d$ and $s$ are the average deficit and surplus of $\PS^*$ and $\NS^*$, respectively, and $d+s \leq 1$.
Denote the average deficit of an intersection of $\ell$ sets in $\PS^*$ by $d_\ell$ and the average surplus of an intersection of $\ell$ sets in $\NS^*$ by $s_\ell$. 

\begin{lemma}
	\label{lem:avg-of-intersect}
	For even $\ell$, $d_\ell+s_\ell \le \frac{5\ell}{2} -2$. For odd $\ell$, $d_\ell+s_\ell \le \frac{5(\ell-1)}{2}+1$.
\end{lemma}


\begin{proof}
		We first prove the following claim: $d_2+s_2\le 3$.
		
		To prove the claim, we begin with a simple observation:
		Recall that $\PS^*$ and $\NS^*$ are complements. For every $P_1,P_2\in\PS^*$ whose complements are $N_1,N_2\in\NS^*$ it holds that
		$f(P_1\cup P_2) = f(\overline{\overline{P_1}\cap \overline{P_2}}) = f(\overline{N_1\cap N_2})\le -f(N_1\cap N_2)+\epsilon+\delta$ (Corollary \ref{cor:comp-val}).
		Similarly,
		$f(N_1\cup N_2) = f(\overline{\overline{N_1}\cap \overline{N_2}}) = f(\overline{P_1\cap P_2})\ge -f(P_1\cap P_2)-\epsilon+\delta$.
		By 1-modularity we thus have	
		$f(P_1\cap P_2)\ge f(P_1)+f(P_2)-f(P_1\cup P_2)-1\ge f(P_1)+f(P_2)+f(N_1\cap N_2)-\delta-2$, and $f(N_1\cap N_2)\le f(N_1)+f(N_2)-f(N_1\cup N_2)+1 \le f(N_1) + f(N_2)+f(P_1\cap P_2)-\delta+2$.
		
		We now take the average over $P_1,P_2$. This gives $M-d_2\ge 2M-2d+(-M+s_2)-\delta-2$. Similarly, the average over $N_1,N_2$ gives $-M+s_2\le -2M+2s+(M-d_2)-\delta+2$, where $M$ is the maximum absolute value of $f$. Subtracting the second inequality from the first we get
		$2M-d_2-s_2\ge 4M-2d-2s-2M+d_2+s_2-4$, so $d_2+s_2\le d+s+2\leq 3$, completing the proof of the claim.
	
	We can now prove Lemma \ref{lem:avg-of-intersect} by induction: It holds for $\ell=1$ and $\ell=2$ using the above claim.
	For every $\ell$ sets $P_1,\dots,P_\ell\in\PS^*$, by 1-modularity the value of their intersection satisfies
$f((P_1\cap \ldots\cap P_{\ell_1})\cap (P_{\ell_1+1}\cap \ldots \cap P_\ell)) \ge f(P_1\cap \ldots \cap P_{\ell_1})+f(P_{\ell_1+1}\cap \ldots \cap P_\ell)-M-1$, where $M$ is the maximum absolute value of $f$. Taking the average over the $\ell$ sets we get $M-d_\ell\ge M-d_{\ell_1}-d_{\ell-\ell_1}-1$. Similarly, $-M+s_\ell \le -M +s_{\ell_1}+s_{\ell-\ell_1}+1$. So $d_\ell+s_\ell\le +d_{\ell_1}+s_{\ell-\ell_1}+d_{\ell_1}+s_{\ell-\ell_1}+2$.	
	
	If $\ell$ is even we can take $\ell_1=2$ and using the induction hypothesis get $d_\ell+s_\ell\le 3+\frac{5(\ell-2)}{2} -2+2=\frac{5\ell}{2} -2$.	
	If $\ell$ is odd we can take $\ell_1=1$ and using the induction hypothesis get $d_\ell+s_\ell\le 1+ \frac{5(\ell-1)}{2} -2 +2=\frac{5(\ell-1)}{2}+1$.	
	This completes the proof.
\end{proof}

\begin{observation}
	\label{obs:item-appearances}
	Every item is contained in $1/2^\ell$ of the intersections of $\ell$ sets in $\PS^*$ or $\NS^*$.
\end{observation}

Observation \ref{obs:item-appearances} allows us to use the main lemma for upper-bounding the Kalton constants (Lemma~\ref{lem:KR}) with $\alpha=1/2^\ell$.
Together with Lemma \ref{lem:avg-of-intersect} that establishes the average deficit and surplus, this enables us to obtain the following result.

\begin{theorem}
\label{thm:K'}
Suppose that for fixed $r,\theta\in \Replus$ there exist $(\alpha,r,\theta)$-expanders. Then the strong Kalton constant satisfies $\Kalton \le \frac{2r + N_2 - \theta}{1 - \theta}$, where if $\alpha = \frac{1}{2}$ then $N_2=-\frac{1}{2}$; if $\alpha = \frac{1}{4}$ then $N_2=\frac{1}{2}$; if $\alpha = \frac{1}{8}$ then $N_2=2$; if $\alpha = \frac{1}{16}$ then $N_2=3$; and if $\alpha = \frac{1}{32}$ then $N_2=4.5$.
\end{theorem}

\begin{proof}
For every $\ell$, by taking the intersections of $\ell$ sets in $\PS^*$ and $\NS^*$ we get collections $\mathcal{G}$ and $\mathcal{G}'$, both with $\frac{1}{2^\ell}$-frequent items (Observation \ref{obs:item-appearances}), whose average deficit and surplus are $d_\ell$ and $s_\ell$, respectively.
Moreover, there exists some $k$ such that $\mathcal{G}$ and $\mathcal{G}'$ have $2k$ sets each, and an expander $G_k$ with the above parameters exists.
We now apply Lemma \ref{lem:UB-main-lemma}, which gives an upper bound of
$\Kalton \le \frac{\frac{1}{2}(d_\ell+s_\ell)+ 2r - 2 }{1-\theta} + 1.$
By Lemma~\ref{lem:avg-of-intersect}, $d_\ell+s_\ell$ is at most $3, 6, 8, 11$ for $\ell=2, 3, 4, 5$, respectively. This implies the theorem.
\end{proof}

To use Theorem~\ref{thm:K'} one needs to substitute in parameters of expanders.
Theorem~\ref{thm:K'} combined with the $(\frac{1}{16},4,\frac{4}{15})$-expanders from Lemma~\ref{lem:expander} implies
that the value of the strong Kalton constant satisfies $\Kalton \le \frac{2\cdot 4+3-0.26}{1-0.26}=14.637$.

\subsection{Strengthening the Upper Bound on $\Kalton$}
\label{appx:UB-strong}

Using our techniques, the upper bound on $\Kalton$ can be improved even further.
The sources of these improvements are twofold.
First, all of our results are derived using bi-regular expanders; better bounds can be obtained using more sophisticated expanders.
Second,
improvements can be obtained by using additional properties of $\epsilon$-modular functions to improve the bounds in Lemma~\ref{lem:avg-of-intersect}.
Improvements of the second type are demonstrated in this section, where the main result is showing that $\Kalton < 12.65$ (establishing Theorem~\ref{thm:upper15}).

\begin{definition}
	Given $0 \le \alpha \le \frac{1}{2}$ and $\mu \ge 0$, an {\em $(\alpha,\mu)$-collection-pair} is a pair of collections $D,S$, such that in both collections items are $\alpha$-frequent, the average value of sets in $D$ is at least $M - d$, the average value of sets in $S$ is at most $-M+s$, and $\frac{1}{2}(d+s)=\mu$.
\end{definition}

For a given $\mu \ge 0$, let $\alpha[\mu]$ denote the smallest $0 \le \alpha \le \frac{1}{2}$ such that there is some $\mu' \le \mu$ for which an $(\alpha,\mu')$-collection-pair exists. By Lemma~\ref{lem:avg-of-intersect}, we may assume for our given 1-modular function $f$ that $\alpha[\frac{1}{2}] \le \frac{1}{2}$, $\alpha[\frac{3}{2}] \le \frac{1}{4}$, and $\alpha[4] \le \frac{1}{16}$.

Fix some $\frac{3}{2} < \delta < 4$ whose value will be optimized later. (Intuitively, we are aiming at $\delta$ satisfying $\alpha[\delta] \simeq \frac{1}{8}$.) Now we consider two cases, each addressed in its own lemma.

\begin{lemma}
	\label{lem:case1alpha}
	Suppose that $\alpha[\delta] \le \frac{\alpha[\frac{3}{2}]}{2}$. Then $\alpha[2\delta - \frac{1}{2}] \le \frac{1}{64}$.
\end{lemma}

\begin{proof}
	Consider the $\mu' \le \delta$ for which an $(\alpha[\delta],\mu')$-collection-pair $(D,S)$ with deficit $d'$ and surplus $s'$ exists, and let $m$ denote the number of sets in $D$ and in $S$ (sets appearing more than once are counted more than once). Consider now the two collection-pairs $(D_\cap,S_\cap)$ and $(D_\cup,S_\cup)$ obtained by taking all $m^2$ pairwise intersections or unions (respectively) of sets from $D,S$. (A pair is generated by picking one set and then another set, with repetitions.) Let $\alpha_{\cap}$ and $\alpha_{\cup}$ be the item frequencies associated with these two collection-pairs, respectively. Let $d_\cup,d_\cap$ be the deficits of $D_\cup, D_\cap$, respectively, and let $s_\cup,s_\cap$ be the surpluses of $S_\cup, S_\cap$, respectively. Let $\mu_\cup$ (resp., $\mu_\cap$) be the average of $d_\cup,s_\cup$ (resp., $d_\cap,s_\cap$). Then:
	\begin{itemize}
		
		\item $\alpha_{\cap} =  (\alpha[\delta])^2 \le  \left(\frac{\alpha[\frac{3}{2}]}{2}\right)^2 \le \frac{1}{64}$;
		
		\item $\alpha_{\cup} < 2\alpha[\delta] \le \alpha[\frac{3}{2}]$.
		
	\end{itemize}
	It follows from the second bullet that $\mu_\cup\ge \frac{3}{2}$. The 1-modularity condition implies that $2d'-d_\cup+1\ge d_\cap$ and $2s'-s_\cup+1\ge s_\cap$, and so by averaging $2\mu'-\mu_\cup+1\ge \mu_\cap$. Substitute $\mu_\cup \geq 3/2$ and $\mu' \leq \delta$ to get $2\delta-\frac{1}{2} \ge \mu_\cap$. By the first bullet, it follows that $\alpha[2\delta - \frac{1}{2}] \le \frac{1}{64}$, as required.
\end{proof}

\begin{lemma}
	Suppose that $\alpha[\delta] > \frac{\alpha[\frac{3}{2}]}{2}$. Then $\alpha[9 - \delta] \le \frac{1}{256}$.
\end{lemma}

\begin{proof}
	Consider the $\mu' \le \frac{3}{2}$ for which an $(\alpha[\frac{3}{2}],\mu')$-collection-pair $(D,S)$ with deficit $d'$ and surplus $s'$ exists. For the collection-pair $(D_\cap,S_\cap)$ (that we shall denote by $(\tilde D,\tilde S)$) the item frequencies are $\left(\alpha[\frac{3}{2}]\right)^2$, and the average of the deficit and surplus $\tilde{\mu}$ is at most $2\mu' + 1 \le 4$ (using 1-modularity). Consider now $({\tilde D}_\cap,{\tilde S}_\cap)$ and $({\tilde D}_\cup,{\tilde S}_\cup)$. Let ${\tilde\alpha}_{\cap}$ and ${\tilde\alpha}_{\cup}$ be the item frequencies associated with these two collections respectively. Then
	\begin{itemize}
		
		\item ${\tilde\alpha}_{\cap} =  \left(\alpha[\frac{3}{2}]\right)^4 \le \frac{1}{256}$,
		
		\item ${\tilde\alpha}_{\cup} < 2\left(\alpha[\frac{3}{2}]\right)^2 < \alpha[\delta]$,
	\end{itemize}
	where the last inequality uses the fact that  $\alpha[\frac{3}{2}] \le \frac{1}{4}$ and the premise of the lemma.
	
	Similarly to the proof of Lemma~\ref{lem:case1alpha}, it follows from the second bullet that $\tilde{\mu}_\cup\ge \delta$. The 1-modularity condition then implies that $9 - \delta \ge 2\cdot 4 - \delta + 1 \ge 2\tilde{\mu} - \tilde{\mu}_\cup + 1 \ge \tilde{\mu}_\cap$. This completes the proof of the lemma.
\end{proof}

\begin{corollary}
	\label{cor:improve}
	For every $\frac{3}{2} < \delta < 4$, every 1-modular function $f$ as above has an $(\alpha,\mu)$-collection-pair with either $\alpha \le \frac{1}{64}$ and $\mu \le 2\delta - \frac{1}{2}$, or $\alpha \le \frac{1}{256}$ and $\mu \le 9 - \delta$.
\end{corollary}

\paragraph{Improved upper bound no.~1: $\Kalton < 13.25$}
We now apply Lemma \ref{lem:UB-main-lemma}, which gives an upper bound of
$\Kalton \le \frac{\mu+ 2r - 2 }{1-\theta} + 1$ when there are $(\alpha,r,\theta)$-expanders. By Corollary \ref{cor:improve} we get
$$
\Kalton \le \max\left\{\frac{2\delta + 2r_1 - 2.5}{1 - \theta_1}, \frac{-\delta + 2r_2+7}{1 - \theta_2}\right\} + 1,
$$
where $r_1,\theta_1$ are the best expander parameters for $\alpha = \frac{1}{64}$,  and  $r_2,\theta_2$ are the best expander parameters for $\alpha = \frac{1}{256}$.
Plugging in $\delta=43/16$, and using Lemma \ref{lem:expander} we get that $\Kalton\le 13.2461$.

\paragraph{Improved upper bound no.~2: $\Kalton < 12.65$}
The $13.25$ bound can be further improved by using Lemma \ref{lem:UB-main-lemma-full} (which is an improved version of Lemma \ref{lem:UB-main-lemma}).
To make use of Lemma \ref{lem:UB-main-lemma-full} we need $d'+s'$ (the lower bounds on the average deficit and surplus of the target sets) to be sufficiently large.
We therefore distinguish between two cases as shown next.

Suppose $f$ has an $(\alpha,\mu)$-collection-pair with $\alpha\leq \frac{1}{64}$ and $\mu\leq 2\delta -1/2$ (the first case in Corollary \ref{cor:improve}), and consider an $(\alpha,3,\frac{3}{11})$-expander (which exists by Lemma \ref{lem:expander}). Consider the target sets of the collection pair guaranteed by Lemma \ref{lem:KR}; their frequency is $\frac{\alpha}{3/11} \leq \frac{1/64}{3/11} < 1/16$. If $d'+s' \leq 5.08$, then we apply Lemma \ref{lem:UB-main-lemma} with the target sets above serving as the source sets, and using the existence of $(\frac{1}{16},4,\frac{4}{15})$-expanders (Lemma \ref{lem:expander}), to get $\Kalton < 12.65$ (by substituting $d'+s' \leq 5.08$, $r=4$, $\theta=4/15$).
Similarly, if $f$ has an $(\alpha,\mu)$-collection-pair with $\alpha\leq \frac{1}{256}$ and $\mu\leq 9-\delta$ (the second case in Corollary \ref{cor:improve}), repeat an analogous analysis using an  $(\alpha,3,\frac{3}{19})$-expander (which exists by Lemma \ref{lem:expander}) to get $\Kalton < 12.65$, as above (this is valid since the frequency of the target sets is $\frac{\alpha}{3/19} \leq \frac{1/256}{3/19} < 1/16$).

So now we apply Lemma \ref{lem:UB-main-lemma-full} in the case where $d'+s' > 5.08$.
By Corollary \ref{cor:improve}, we get that for every $\frac{3}{2} < \delta < 4$ it holds that $\Kalton$ is upper-bounded by
$$
\max\left\{\frac{2\delta-\frac{1}{2} - \frac{5.08 \theta_1}{2} + 2(r_1 - 1)}{1 - \theta_1}, \frac{9-\delta - \frac{5.08 \theta_2}{2} + 2(r_2-1)}{1 - \theta_2}\right\} + 1,
$$
where $r_1,\theta_1$ are the best expander parameters for $\alpha = \frac{1}{64}$,  and  $r_2,\theta_2$ are the best expander parameters for $\alpha = \frac{1}{256}$.
Plugging in $\delta=43/16$, and using Lemma \ref{lem:expander} we get that $\Kalton \le 12.622$.
Together, we get that $\Kalton < 12.65$.

\paragraph{Remark} 
The bounds shown in this section illustrate the techniques we use.
Clearly, these techniques can be extended to give even better bounds by, e.g., establishing better expanders and applying our ideas recursively. We save these extensions for future work.

\section{Improved Lower Bounds}
\label{sec:lower}
In this section we prove the lower bound stated in Theorem \ref{thm:lower1} on the strong Kalton constant; i.e., we prove that $\Kalton \geq 1$. The lower bound stated in Theorem \ref{thm:UpperLowerWeak} on $\Kalweak$ is proved in Appendix \ref{sub:LB-weak}.

We begin by explicitly constructing a set function $f$ that establishes the lower bound. The idea underlying $f$'s construction is to achieve symmetry properties which facilitate its analysis; this idea is developed in Section \ref{sub:k-M-symm}, and the analysis appears in Section \ref{sub:claims}.

\begin{proof}[Proof of Theorem \ref{thm:lower1}]
	The following set function with $n=70$ is 2-modular and tightly 2-linear, thus showing $\Kalton\ge 1$.
	Let $M=2$. The positive support $\PS$ and the negative support $\NS$ of $f$ (Definition \ref{def:PS-NS}) are as follows:
	There are $8$ sets in $\PS$, each with 35 out of the 70 items, such that every item appears in exactly $4$ of the $8$ sets (notice that ${8 \choose 4}=70$). (We can fix any such collection of sets as $\PS$ without loss of generality.) The negative support $\NS$ is the complement collection of $\PS$.
	Notice that $\PS,\NS$ support uniform distributions $P^+,P^-$ with equal item marginals, as required.
	The values of sets under $f$ are determined by the following rules, where the first applicable rule applies (and hence $f$ is well defined):
\begin{enumerate}

\item  Each positive support set $S \in \PS$ has value $f(S) = M = 2$.

\item For every two sets $S_1$ and $S_2$ in $\PS$ and every set $R$, we have $f(S_1 \cup (S_2 \cap R)) = 1$,
and likewise 
$f(S_1 \cap (S_2 \cup R)) = 1$. (In particular, for every two sets $S_1,S_2 \in \PS$, $f(S_1 \cup S_2) =  f(S_1 \cap S_2) = 1$.)

\item 
We impose $-f(S) = f(\bar{S})$ and derive from this sets with negative value.

\item All other sets have value~0.

\end{enumerate}

\begin{claim}
	\label{cla:tight-2-linear}
	$f$ is tightly 2-linear.
\end{claim}

\begin{claim}
	\label{cla:2-modular}
	$f$ is 2-modular.
\end{claim}

The proofs of Claims \ref{cla:tight-2-linear} and \ref{cla:2-modular} use the tool of $(k,M)$-symmetric set functions introduced in Section \ref{sub:k-M-symm}. In particular, the proof of the latter claim is based on an analysis of a subset of selected cases, which are proven to be sufficient to establish $\epsilon$-modularity for $f$ due to its symmetry properties (see Lemma \ref{lem:shortcaseanalysis}).
The proofs of the claims appear in Section \ref{sub:claims}, thus completing the proof of Theorem \ref{thm:lower1}. 
\end{proof}

\subsection{$(k, M)$-Symmetric Set Functions}
\label{sub:k-M-symm}

In this section we introduce the class of $(k, M)$-symmetric set functions, which are tightly $M$-linear (Proposition \ref{pro:distM}), while enjoying many symmetries. In general, checking whether a set function over $n$ items is approximately modular involves verifying that roughly $2^{2n}$ modular equations (one equation for every pair of sets) approximately hold. Verifying approximate modularity for $(k, M)$-symmetric set functions becomes an easier task thanks to their symmetries.
We begin by introducing some terminology.

\begin{definition}
	\label{def:generating}
	A collection $\mathcal{G} = \{S_1, S_2, \ldots\ , S_g\}$ of $g$ subsets of $U$ is {\em generating} if the following conditions hold:
	
	\begin{enumerate}
		
		\item It is \emph{covering}, i.e., $\bigcup_{S_j \in \mathcal{G}} S_j = U$  (every item is contained in at least one set).
		
		\item It is \emph{item-differentiating}, i.e., $\bigcap_{S_j \in \mathcal{G} \mid i\in S_j} S_j = \{i\}$ for every item $i$ (equivalently, for every pair of items, there is a set containing one but not the other). Note that this implies in particular that $\bigcap_{S_j \in \mathcal{G}} S_j = \emptyset$ (no item is contained in all sets).
		
	\end{enumerate}
\end{definition}

Observe that if a collection $\mathcal{G}$ is generating, then every subset of $U$ can be generated from sets in $\mathcal{G}$ by a sequence of intersections and unions (possibly, in more than one way).
Also observe that given a generating collection $\mathcal{G} = \{S_1, S_2, \ldots\ , S_g\}$, the \emph{complement collection} $\bar{\mathcal{G}} = \{\bar{S_1}, \bar{S_2}, \ldots\ , \bar{S_g}\}$ (obtained by complementing each of the generating sets) is also generating. This can be shown by applying De Morgan's laws.

\begin{definition}
	\label{def:canonical-collect}
	A generating collection $\mathcal{G}$ is {\em canonical} if the number $g$ of generating sets is even (we denote $g = 2k$ for some positive integer $k$), every item is contained in exactly $k$ sets from $\mathcal{G}$, every set in $\mathcal{G}$ contains exactly $n/2$ items, and $n={2k\choose k}$.
\end{definition}

Given a canonical generating collection $\mathcal{G}$ where $\mathcal{G} = \{S_1, S_2, \ldots\ , S_{2k}\}$, items can be thought of as balanced vectors in $\{\pm 1\}^{2k}$, where coordinate $j$ of item $i$ is $+1$ if $i \in S_j$, and $-1$ if $i \in \bar{S_j}$. Observe also that if $\mathcal{G}$ is a canonical generating collection, then so is its complement~$\bar{\mathcal{G}}$.

A {\em generating circuit} $C$ is a directed acyclic graph (namely, with no directed cycles) with $g$ nodes referred to as {\em input nodes} (these nodes have no incoming edges), one node referred to as the {\em output node} (this node has no outgoing edges), and in which each non-source node has at most two incoming edges. Nodes with two incoming edges are labeled by either a $\cap$ (intersection) or $\cup$ (union) operation, whereas nodes with one incoming edge are labeled by $\bar{\cdot}$ (complementation). Associating the input nodes with the $g$ sets of a generating collection $\mathcal{G}$, the set at each node is computed by applying the respective operation on the incoming sets (either intersection, union, or complementation), and the output of the circuit is the set at the output node.

Given a circuit $C$, the {\em dual circuit} $\hat{C}$ is obtained by replacing $\cap$ by $\cup$ and vice versa. For a given generating collection $\mathcal{G}$ and a permutation that maps $\mathcal{G}$ to the input nodes of circuit $C$, if $S$ is the set output by $C$, we refer to the set output by the dual circuit $\hat{C}$ as the {\em dual} of $S$, and denote in by $\hat{S}$.
It can be shown that given $\mathcal{G}$ and $S$, the dual set $\hat{S}$ is well defined, in the sense that for all circuits $C$ that generate $S$, their duals generate the same $\hat{S}$. (In a canonical generating collection $\mathcal{G}$, every item appears in exactly $k$ generating sets. This induces a perfect matching over items, where two items are matched if there is no generating set in which they both appear, or equivalently, if the vectors in $\{\pm 1\}^{2k}$ representing them are negations of each other. Given a set $S$, its dual can be seen to be the set $\hat{S}$ that contains all items that are {\em not} matched to items in $S$. In particular, every set in $\mathcal{G}$ is the dual of itself.)

From now on we restrict attention to set functions $f$ whose closest linear set function is the zero function.
Recall from Definition \ref{def:PS-NS} that the positive and negative supports $\PS$ and $\NS$ of~$f$ are the collections of sets that $f$ assigns maximum or minimum values to and are in the supports of distributions $P^+$ or $P^-$, respectively (see Lemma~\ref{lem:charac} above).

\begin{definition}
	\label{def:(k,M)symmetric}
	For integers $k \ge 2$ and $M \ge 1$, we say that a set function $f$ over a set $U$ of $n = {2k \choose k}$ items is $(k,M)$-symmetric if it has the following properties:
	\begin{enumerate}
		
		\item {\em Integrality:} $f$ attains only integer values.
		
		\item {\em Antisymmetry:} for every set $S$ and its complement $\bar{S}$ it holds that $f(S) = -f(\bar{S})$.
		
		\item {\em Canonical generating sets:} positive support $\PS$ contains $2k$ sets $P_1, \ldots, P_{2k}$ (of value $M$), negative support $\NS$ contains $2k$ sets $N_1 = \bar{P}_1, \ldots, N_{2k}=\bar{P}_{2k}$ (of value $-M$, these are the complements of the sets in $\PS$), and $\PS$ (and likewise $\NS$) is a canonical generating collection.
		
		\item {\em Generator anonymity:} let $C$ be an arbitrary generating circuit. Then the value (under $f$) of the set output by the circuit is independent of the permutation that determines which of the generating sets from $\PS$ is mapped to which source node. 
		
		\item {\em Dual symmetry:} for every set $S$ and its dual $\hat{S}$ it holds that $f(S) = f(\hat{S})$.
		
	\end{enumerate}
\end{definition}

\begin{proposition}
	\label{pro:distM}
	Every $(k,M)$-symmetric function $f$
	is tightly $M$-linear, and the 0 function is a linear function closest to $f$.
\end{proposition}

\begin{proof}
	Consider item~3 (canonical generating sets) in Definition~\ref{def:(k,M)symmetric}. By the virtue of $\PS$ being a canonical generating set, the uniform distribution over $\PS$  has marginal $1/2$ for every item in $U$, and likewise for $\NS$. Hence Lemma~\ref{lem:charac} implies that the~0 function is a linear function closest to $f$. As the maximum value of $f$ is $M$, it then follows that $f$ is tightly $M$-linear.
\end{proof}

We shall design certain $(k,M)$-symmetric functions and would like to prove that they are $\epsilon$-modular, typically for $\epsilon = 2$. The case analysis involved in checking $\epsilon$-modularity can be reduced to the cases outlined in the following lemma.

\begin{lemma}
	\label{lem:shortcaseanalysis}
	To verify that a $(k,M)$-symmetric function $f$ is $\epsilon$-modular it suffices to verify the approximate modularity equation
	$|f(S) + f(T) - f(S\cap T) - f(S\cup T)| \le \epsilon$
	in the cases where $S$ and $T$ satisfy all of the following conditions: (a) $|f(T)| \le |f(S)|$;  (b) $0 \le f(S) \le M$; and (c) $f(S\cap T) \le f(S\cup T)$.
\end{lemma}

\begin{proof}
	Suppose that we wish to verify the approximate modularity condition $|f(S) + f(T) - f(S \cap T) - f(S \cup T)| \le \epsilon$ for two sets $S$ and $T$. If $S$ and $T$ satisfy the three conditions in the lemma, then indeed the approximate modularity condition will be checked directly. Hence it remains to show that even if $S$ and $T$ do not satisfy some of the conditions of the lemma, there will be two other sets, say $S'$ and $T'$, that do satisfy the conditions, and such that the approximate modularity condition holds for $S'$ and $T'$ if and only if it holds for $S$ and $T$.
	
	Suppose that condition (a) is violated, namely, $|f(T)| \ge |f(S)|$. Then simply interchange $S$ and $T$ and then condition (a) holds.

	Suppose that condition (a) holds and condition (b) is violated, namely, $f(S) < f(T)$. Then consider the complement sets $\bar{S}$ and $\bar{T}$, for which both conditions (a) and (b) hold. By anti symmetry they satisfy $f(S) = -f(\bar{S})$, $f(T) = -f(\bar{T})$, $f(S\cap T) = -f(\bar{S}\cup \bar{T})$, and  $f(S\cup T) = -f(\bar{S}\cap \bar{T})$. Hence the approximate modularity condition for $\bar{S}$ and $\bar{T}$ implies it for $S$ and $T$.
	
	Suppose that conditions (a) and (b) hold, and condition (c) is violated, namely, $f(S\cap T) > f(S\cup T)$.  In this case, consider the dual sets $\hat{S}$ and $\hat{T}$.  By dual symmetry we have that $f(S) = f(\hat{S})$, $f(T) = f(\hat{T})$, $f(S\cap T) = f(\hat{S}\cup \hat{T})$, and  $f(S\cup T) = f(\hat{S}\cap \hat{T})$. Hence the approximate modularity condition for $\hat{S}$ and $\hat{T}$ implies it for $S$ and $T$. Moreover, for $\hat{S}$ and $\hat{T}$ conditions (a), (b) are inherited from $S$ and $T$, and condition (c) does hold.
\end{proof}

In the remainder of the section we use $(k,M)$-symmetric functions to prove Claims \ref{cla:tight-2-linear} and \ref{cla:2-modular}.

\subsection{Proofs of Claims \ref{cla:tight-2-linear} and \ref{cla:2-modular}} 
\label{sub:claims}

\begin{proof}[Proof of Claim \ref{cla:tight-2-linear}]
	It is not hard to verify that the function $f$ defined above is  $(k,M)$-symmetric according to Definition~\ref{def:(k,M)symmetric} for $k=4$ and $M=2$:
	the main thing to check is that there are no two positive sets that are complements of each other (and consequently antisymmetry is enforced by rule~3), and this is implied by the proof of Claim~\ref{claim:containment} below. Proposition~\ref{pro:distM} implies that $f$ is tightly 2-linear. 
\end{proof}

We say that $S$ is a \emph{positive set} if $f(S) > 0$, a {\em negative set} if $f(S) < 0$, and a {\em zero set} if $f(S) = 0$.

\begin{claim}
	\label{claim:containment}
	There is no pair of sets $S$ and $T$ such that one of them is positive and the other is negative and $S \subset T$.
\end{claim}

\begin{proof}
	We show that no positive set is contained in a negative set. The opposite direction can be shown analogously.
	It is sufficient to show that no minimal positive set is contained in a maximal negative set.
	View each item as a balanced vector in $\{\pm 1\}^8$.
	Let $S$ be a minimal positive set. $S$ is the intersection of two sets in PS, thus has two coordinates (out of eight), say 1 and 2, fixed to $+1$ and contains all items that agree with both.
	Let $T$ be a maximal negative set. $T$ is the union of two sets in NS, thus has two coordinates, say 3 and 4, fixed to $-1$ and contains all vectors that agree with at least one of the coordinates.
	The item corresponding to the vector that has $+1$ on all these four coordinates (here is where we used the fact that $k \ge 4$) is in $S$ but not in $T$.
\end{proof}

\begin{proof}[Proof of Claim \ref{cla:2-modular}]
	We show that $f$ is 2-modular. Consider two sets $S$ and $T$. By Lemma~\ref{lem:shortcaseanalysis}, the cases in the following case analysis suffices in order to establish 2-modularity.
	
	\begin{enumerate}
		
		\item $f(S) = 2$ (namely, $S \in \PS$).  In this case both $f(S\cup T) \ge 0$ and $f(S\cap T) \ge 0$, by Claim~\ref{claim:containment}, a fact that will be implicitly used in the subcases below.
		
		\begin{enumerate}
			
			\item $f(T) = 2$ (namely, $T \in \PS$). In this case $f(S\cup T)=f(S\cap T) = 1$, by definition of $f$ (rule~2). Thus, $f(S) + f(T) - f(S\cup T) - f(S\cap T) = 2$.
			
			\item $f(T) = 1$.  In this case $f(S)+f(T)=3$. So it suffices to show that either $f(S\cup T) \ge 1$ or $f(S \cap T) \geq 1$. By the definition of $f$, there are two possibilities:
			
			\begin{enumerate}
				
				\item $T$ is contained in some $\PS$ set $S_2$, and then $f(S \cup T) \ge 1$ by the condition $f(S_1 \cup (S_2 \cap R)) =  1$ (with $S_1$ serving as $S$).
				
				\item  $T$ contains some $\PS$ set $S_2$, and then $f(S \cap T) \ge 1$ by the condition $f(S_1 \cap (S_2 \cup R)) = 1$.
				
			\end{enumerate}
			
			\item $f(T) = 0$. 
			In this case $f(S) + f(T) = 2$ and  $0 \le f(S \cup T) + f(S \cap T) \le 4$, satisfying the 2-modularity condition.
			
			\item $f(T) < 0$. Given that $f(S) > 0$ and $f(T) < 0$ we have that $f(S\cup T) = f(S\cap T) \ge 0$, by Claim~\ref{claim:containment}. Hence the stronger 1-modularity condition holds.
			
		\end{enumerate}
		
		\item $f(S) = 1$. Both $f(S\cup T) \ge 0$ and $f(S\cap T) \ge 0$, by Claim~\ref{claim:containment}.
		
		\begin{enumerate}
			
			\item $f(T) = 1$.  The 2-modularity condition then holds since $f(S)+f(T)=2$ and  $0 \le f(S \cup T) + f(S \cap T) \le 4$.
			
			\item $f(T) = 0$. 
			Observe that it cannot be that both  $f(S \cup T) = 2$ and  $f(S \cap T)=2$.
			Hence $f(S) + f(T) = 1$ and  $0 \le f(S \cup T) + f(S \cap T) \le 3$, satisfying the 2-modularity condition.
			
			\item $f(T) = -1$. Given that $f(S) > 0$ and $f(T) < 0$ we have that $f(S\cup T) = f(S\cap T) \ge 0$, by Claim~\ref{claim:containment}. Hence $f(S) + f(T) =  f(S \cup T) + f(S \cap T)$ in this case.
			
		\end{enumerate}
		
		\item $f(S) = f(T) = 0$. We need to show that $-2 \leq f(S \cup T) + f(S \cup T) \leq 2$. This might be violated only if $\max[| f(S \cup T)|,  |f(S \cup T)|] = 2$. Lemma~\ref{lem:shortcaseanalysis} implies that it suffices to consider the case that the set $f(S\cup T) = 2$. Namely, $S\cup T \in \PS$. Then necessarily $|S \cap T| < k$ and $0 \leq f(S \cap T) < 2$. We show that $f(S \cap T) \neq 1$, and then indeed $-2 \leq f(S \cup T) + f(S \cup T) \leq 2$.
		
		Given that  $|S \cap T| < k$, the only rule that might cause $f(S \cap T) = 1$ is that $S \cap T = S_1 \cap (S_2 \cup R)$, where $S_1, S_2 \in \PS$. But given that $k \ge 4$ and that $S \cup T$ (which we shall call $P$) is also in $\PS$, it follows that either $S_1$ or $S_2$ are equal to $P$. (No set in $\PS$ contains the intersection of two other sets from $\PS$.) If $S_1 = P$ then $S$ (and also $T$) is sandwiched between $S_1 \cap (S_2 \cup R) \subset S \subset S_1$, and hence is itself of the form $S = S_1 \cap (S_2 \cup R')$, contradicting the assumption that $f(S) = 0$. If $S_2 = P$ then $S$ (and also $T$) is sandwiched between $S_1 \cap (S_2 \cup R) \subset S \subset S_2$, and this implies that without loss of generality $R = \emptyset$. Hence $S$ is of the form $S = (S_1 \cup R') \cap S_2$, contradicting the assumption that $f(S) = 0$.
	\end{enumerate}
	
	Hence we established that $f$ is 2-modular.
\end{proof}

\section{Learning $\Delta$-Linear and $\epsilon$-Modular Functions}
\label{sec:learning}
Consider the following natural question: We have value-query access to a $\Delta$-linear set function $f$. Our goal is to learn in polynomial-time a ``hypothesis'' linear set function $h$ that is $\delta$-close to $f$. How small can $\delta$ be as a function of $\Delta$ and $n$?
We say that an algorithm \emph{$\delta$-learns} a set function $f$ if given value-query access to $f$ it returns in polynomial-time a linear set function that is $\delta$-close to $f$.

The work of Chierichetti et al.~\cite{CDDK} (Theorem 4) presents an algorithm that $O(\epsilon\sqrt{n})$-learns \emph{$\epsilon$-modular} set functions. Since every $\Delta$-linear set function is $4\Delta$-modular (Proposition \ref{pro:linear-to-modular}), their algorithm also $O(\Delta\sqrt{n})$-learns $\Delta$-linear set functions.
The algorithm of \cite{CDDK} is randomized, and after making $O(n^2\log n)$ non-adaptive queries to the function it is learning, returns a $\delta$-close linear set function with probability $1-o(1)$.

In this section we present an alternative algorithm for $O(\Delta\sqrt{n})$-learning $\Delta$-linear set functions (and since every $\epsilon$-modular set function is $\Kalton\epsilon$-linear for a constant $\Kalton$, also for $O(\epsilon\sqrt{n})$-learning $\epsilon$-modular set functions).
Our algorithm (Algorithm \ref{alg:learning}) is simple, \emph{deterministic}, and makes a \emph{linear} number of non-adaptive queries to the function it is learning. 

\subsection{Learning with the Hadamard Basis}

A \emph{Hadamard basis} is an orthogonal basis $\{v_1, \ldots, v_n\}$ of $\mathbb{R}^n$ which consists of vectors in $\{\pm 1\}^n$. Let $\{y_1, \ldots, y_n\}=\{v_1/\sqrt{n}, \ldots, v_n/\sqrt{n}\}$, then $\{y_1, \ldots, y_n\}$ is an orthonormal basis of $\mathbb{R}^n$.
Easy recursive constructions of a Hadamard basis are known whenever $n$ is a power of~2 (and also for some other values of $n$).
For all values of $n$ for which a Hadamard basis exists, and for every choice of one particular vector $v \in \{\pm 1\}^n$, we may assume that $v$ is a member of the Hadamard basis. This can be enforced by taking the first vector $v_1$ in an arbitrary Hadamard basis of $\mathbb{R}^n$, and flipping -- in all vectors of the basis -- those coordinates in which $v$ and $v_1$ do not agree.

For every set of items $S\subseteq [n]$, let $v_S\in \{0,1\}^n$ be the indicator vector of $S$. Vector $v_S$ can be written as a linear combination $\sum_{i=1}^n \lambda_i y_i$ of the basis vectors $\{y_1, \ldots, y_n\}$. We shall make use of the following property of orthonormal bases: Express $v_S$ as $\sum_{i=1}^n \lambda_i y_i$; since $\{y_1, \ldots, y_n\}$ is orthonormal, then by the Pythagorean theorem generalized to Euclidean spaces,
\begin{equation}
\sum_{i=1}^n (\lambda_i)^2 = |S|.\label{eq:Had}
\end{equation}

Given a basis vector $v_i$ we denote by $S_i$ the set of indices in which $v_i$ has entries~$+1$, and by $\bar{S_i}$ the set of indices in which $v_i$ has entries~$-1$.

Throughout this section, $i\in[n]$ specifies the index of a basis vector (as in $v_i$ or $y_i$), whereas to index other objects (such as coordinates of vectors, or items) we shall use $j$ rather than $i$. When a vector is indexed by a set (such as $v_S$, or $v_{\{j\}}$) then the vector is the indicator vector (in $\{0,1\}^n$) for the set.

\begin{claim}[Expressing a linear set function using $S_i,\bar{S_i}$]
	\label{cla:val-of-S}
	Consider a linear set function $g$ where $g(S)=\sum_{j\in S}g_j + g(\emptyset)$ for every set $S\subseteq[n]$. Then for every  $v_S$  it holds that $g(S)=\frac{1}{\sqrt{n}}\sum_{i=1}^n \lambda_i (g(S_i)-g(\bar{S_i})) + g(\emptyset)$, where here
$\lambda_1, \ldots , \lambda_n$ are the unique coefficients for which $v_S=\sum_{i=1}^n \lambda_i y_i$. In particular $g_j=g(\{j\})$ can be obtained by substituting in the unique $\lambda_1, \ldots , \lambda_n$ for which
$v_{\{j\}}=\sum_{i=1}^n \lambda_i y_i$.
\end{claim}

\begin{proof}
Let $g'$ be the additive function defined as $g'(S) = g(S) - g(\emptyset)$. Extend the domain of $g'$ from $\{0,1\}^n$ to $\mathbb{R}^n$, giving the additive function $\tilde{g}$ over $\mathbb{R}^n$ defined as follows: $\tilde{g}(x)=\sum_{j\in S} g_jx_j$ for every $x\in \mathbb{R}^n$. Observe that for every Hadamard basis vector $v_i$, $\tilde{g}(v_i) = g(S_i) - g(\bar{S_i})$ (by definition of the sets $S_i,\bar{S_i}$).
By additivity of $\tilde{g}$, for every set $S$ whose  indicator vector $v_S$ can be expressed as $\sum_{i=1}^n \lambda_i y_i$, we have that $\tilde{g}(v_S) = \sum_{i=1}^n \lambda_i \tilde{g}(y_i) = \frac{1}{\sqrt{n}}\sum_{i=1}^n \lambda_i \tilde{g}(v_i) = \frac{1}{\sqrt{n}}\sum_{i=1}^n \lambda_i (g(S_i)-g(\bar{S_i}))$. Finally, $g(S)-g(\emptyset)=g(S)-g_0=\tilde{g}(v_S)$ for every set $S$.
\end{proof}

\begin{lemma}
	\label{lem:differ-in-Hadamard}
	Let $h$ and $g$ be two linear set functions such that for every $i\in[n]$, $h(S_i)-h(\bar{S_i})$ is $O(\Delta)$-close to $g(S_i)-g(\bar{S_i})$, and for every $S'\in\{\emptyset,U\}$, $h(S')$ is $O(\Delta)$-close to $g(S')$.
	Then for every set $S\subseteq [n]$, $h(S)$ is $O(\Delta(1 + \sqrt{\min\{|S|,n-|S|\}}))$-close to $g(S)$.
\end{lemma}

\begin{proof}
The lemma clearly holds for $S\in\{\emptyset,U\}$, since in this case $O(\Delta(1 + \sqrt{\min\{|S|,n-|S|\}}))=O(\Delta\cdot1)$.
For the remainder of the proof assume $0<|S|<n$.
By Claim~\ref{cla:val-of-S} applied to $h$ and $g$,
$h(S)=\frac{1}{\sqrt{n}}\sum_{i=1}^n \lambda_i (h(S_i)-h(\bar{S_i})) + h(\emptyset)$, and
$g(S)=\frac{1}{\sqrt{n}}\sum_{i=1}^n \lambda_i (g(S_i)-g(\bar{S_i})) + g(\emptyset)$.
For every $i\in[n]$, $h(S_i)-h(\bar{S_i})$ estimates $g(S_i)-g(\bar{S_i})$ up to an additive error of $O(\Delta)$, and $h(\emptyset)$ estimates $g(\emptyset)$ up to an additive error of $O(\Delta)$. Thus we get that $h(S)$ estimates $g(S)$ up to an additive error of $\frac{O(\Delta)}{\sqrt{n}}\sum_{i=1}^n |\lambda_i|+O(\Delta)$. Invoking \eqref{eq:Had}, this error is maximized when $(\lambda_i)^2 =|S|/n$ for every $i$. We conclude that an upper bound on the additive estimation error is
\begin{equation}
\frac{O(\Delta)}{\sqrt{n}}n\sqrt{\frac{|S|}{n}} +O(\Delta)= O(\Delta) \sqrt{|S|}+O(\Delta)=O(\Delta\sqrt{|S|}).\label{eq:S-smaller}
\end{equation}
Since $h(S)=h(\emptyset)+h(U)-h(\bar{S})$ and $g(S)=g(\emptyset)+g(U)-g(\bar{S})$, and since we know from \eqref{eq:S-smaller} that $h(\bar{S})$ estimates $g(\bar{S})$ up to an additive error of $O(\Delta\sqrt{|\bar{S}|})=O(\Delta\sqrt{n-|S|})$, we also get that the maximum additive estimation error is
\begin{equation}
O(\Delta)+O(\Delta\sqrt{n-|S|})=O(\Delta\sqrt{n-|S|}).\label{eq:comp-S-smaller}
\end{equation}
Taking the minimum among \eqref{eq:S-smaller} and \eqref{eq:comp-S-smaller} completes the proof.
\end{proof}

\subsection{The Algorithm}

\begin{algorithm}
	{\bf Input:} Value-query access to a set function $f$.
	
	{\bf Output:} A linear set function $h$ where $h(S)=h_0 + \sum_{j\in S} h_j$.
		
	\smallskip
	
	\% Query linearly many values of $f$
	
	query $f(\emptyset)$
	
	{\bf for every} $i\in [n]$ {\bf do} \quad \quad \% Consider the $i$th Hadamard basis vector $v_i$
	
	\quad query $f(S_i)$ and $f(\bar{S_i})$ \quad \% $S_i,\bar{S_i}$ are sets (possibly empty) of the $+1,-1$ entries of $v_i$
	
	{\bf end for}
	
	\smallskip
	
	\% Compute $h$
	
	set $h_0=f(\emptyset)$

	{\bf for every} $j\in [n]$ {\bf do}

	\quad express $v_{\{j\}}$ as $\sum_{i=1}^n \lambda_i y_i$ \quad \quad \% $y_i$ is the Hadamard basis vector $v_i$ after normalization
	
	\quad set $h_j=\frac{1}{\sqrt{n}}\sum_{i=1}^n \lambda_i (f(S_i)-f(\bar{S_i}))$

	{\bf end for}
	\caption{Using the Hadamard basis to learn a linear set function.}
	\label{alg:learning}
\end{algorithm}

For simplicity, we state the following theorem for values of $n$ that are a power of 2. In Remark~\ref{rem:non-power-2} we explain how to extend it beyond powers of 2. In the following Theorem, when referring to Algorithm~\ref{alg:learning}, we mean Algorithm~\ref{alg:learning} run with the first Hadamard basis vector $v_1$ being the all-ones vector.

\begin{theorem}
	\label{thm:hadamard}
	Let $n$ be a power of~2. Given value-query access to a $\Delta$-linear set function $f$ over $n$ items, Algorithm~\ref{alg:learning} returns in polynomial time a linear set function $h$ such that for every set $S\subseteq [n]$, $h(S)$ is $O(\Delta(1 + \sqrt{\min\{|S|,n-|S|\}}))$-close to $f(S)$. Algorithm~\ref{alg:learning} thus $O(\Delta\sqrt{n})$-learns $\Delta$-linear set functions over $n$ items.
\end{theorem}

\begin{proof}
	First note that Algorithm \ref{alg:learning} clearly runs in polynomial time.
By the construction of $h$ in Algorithm \ref{alg:learning} and by Claim \ref{cla:val-of-S}, $h$ is the function that assigns to every set $S=\sum_{i=1}^n \lambda_i y_i$ the value $h(S)=\frac{1}{\sqrt{n}}\sum_{i=1}^n \lambda_i (f(S_i)-f(\bar{S_i})) + f(\emptyset)$.
So for every $i\in[n]$, $h(S_i)-h(\bar{S_i})=f(S_i)-f(\bar{S_i})$, and $h(\emptyset)=f(\emptyset)$.
Let $g$ be a linear set function $\Delta$-close to~$f$, then
for every $i\in[n]$, $h(S_i)-h(\bar{S_i})$ is $2\Delta$-close to $g(S_i)-g(\bar{S_i})$, and $h(\emptyset)$ is $\Delta$-close to $g(\emptyset)$. Recall that $v_1$ is the all-ones vector, which is the indicator vector of $U$. So $h(U)$ is also $2\Delta$-close to $g(U)$.
Invoking
Lemma \ref{lem:differ-in-Hadamard}, we get that for every set $S$, $h(S)$ is $O(\Delta(1 + \sqrt{\min\{|S|,n-|S|\}}))$-close to $g(S)$. Since $g$ is $\Delta$-close to $f$, this completes the proof.
\end{proof}

Theorem~\ref{thm:hadamard} is essentially the best possible, in the following strong sense: Corollary~23 of \cite{CDDK} shows that no algorithm (deterministic or randomized) that performs polynomially-many value queries can find a $o(\Delta\sqrt{n/\log n})$-close linear set function, even if the queries are allowed to be adaptive.
We include a proof sketch of this tightness result in Appendix \ref{appx:learning} 
for completeness. Our tightness proof holds even for learning \emph{monotone} $\Delta$-linear set functions.

\begin{remark}[LP-based approach]
	\label{rem:learning-LP}
	We describe an alternative way to derive the linear function $h$: After querying $f(S)$ for $S=\emptyset$ and for $S=S_i,\bar{S_i}~\forall i\in[n]$, one solves a linear program (LP). The $n+1$ variables $x_j$ of the LP are intended to have the value $g_j$ for every $j$. The constraints are
	$f(S) - \Delta \le \sum_{j\in S} x_j \le f(S) + \Delta$ for every set $S$ that was queried.
	Let $h$ be the linear set function defined by $c_j = x^*_j$ for every $j$ where $x^*$ is a feasible solution of the LP. We claim that $h$ is $O(\Delta(1 + \sqrt{\min\{|S|,n-|S|\}}))$-close to $f$: On each of the queried sets $S$, the values of $h$ and $g$ differ by at most $2\Delta$.
	Since we may assume that $v_1$ is the all-ones vector, the set $S=U$ is one of the queried sets. Invoking
	Lemma \ref{lem:differ-in-Hadamard} and using that $g$ is $\Delta$-close to $f$ shows the claim.
	The advantage of this LP-based approach is that additional constraints can easily be incorporated once more data of $f$ is collected, potentially leading to better accuracy. Likewise, one can easily enforce desirable properties such as non-negativity on $h$ (if $g$ is nonnegative).
\end{remark}

\begin{remark}[Beyond powers of 2]
	\label{rem:non-power-2}
	Using either Algorithm \ref{alg:learning} or the LP-based algorithm described in Remark \ref{rem:learning-LP}, we can $O(\Delta\sqrt{n})$-learn $\Delta$-linear set functions over $n$ items even when $n$ is not a power of 2, as follows. Extend $f$ to $f'$ over $n'\ge n$ items $U'$, where $n'$ is a power of 2 and $U\subseteq U'$, by setting $f'(S)=f(S\cap U)$ for every set $S\subseteq U'$.
	Extend $g$ to $g'$ over $U'$ in the same way.
	Notice that the extended versions $f',g'$ are still $\Delta$-close.
	
	Given $h'$ over $U$ returned by Algorithm \ref{alg:learning}, we define $h$ over $U$ by setting $h(S)=h'(S)$ for every set $S\subseteq U$. The proof of Theorem \ref{thm:hadamard} holds verbatim with the single following modification: the Hadamard basis vector $v_1\in \mathbb{R}^{n'}$  is no longer the all-ones vector, but rather the vector that is +1 on the first $n$ coordinates and $-1$ on the $n' - n$ auxiliary variables. This ensures that we can apply Lemma \ref{lem:differ-in-Hadamard} to $h$ instead of $h'$.
	
	For the LP-based algorithm, we can formulate the LP with $n'+1$ variables, but since we know there is a feasible solution in which $x_j=0$ for every $j>n$ (namely $x=g$), we can add these constraints so that the resulting linear function $h$ is over $n$ items.
\end{remark} 

\appendix

\section{Missing Proof from Section \lowercase{\ref{sec:prelim}} (Preliminaries)}
\label{appx:prelim}

\begin{proof}[Proof of Theorem \ref{thm:pawlik}]
	Let $\epsilon=1$, and let $f$ be the function given in Pawlik's paper ``Approximately Additive Set Functions'' \cite{Pawlik} over the items $X \uplus Y$. We present the function $f$ and the proof for completeness (and due to typos and brevity in \cite{Pawlik}). Let $f(S \cup T)$ denote the value of the union of a set $S \subseteq X$ and a set $T \subseteq Y$. Let $X'$ (resp., $Y'$) denote a non-empty proper subset of $X$ (resp., $Y$). The function $f$ is defined as follows: $f(\emptyset \cup \emptyset)=0$, $f(X' \cup \emptyset)=1$, $f(X \cup \emptyset)=3$, $f(\emptyset \cup Y')=-1$, $f(X' \cup Y')=0$, $f(X \cup Y')=1$, $f(\emptyset \cup Y)=-3$, $f(X' \cup Y)=-1$, $f(X \cup Y)=0$.
	
	As claimed in \cite{Pawlik}, $f$ is weakly $1$-modular. We observe that $f$ is $2$-modular, but no better than $2$ modular (consider two non-disjoint proper subsets of $X$ whose union is $X$).
	
	Let $\mu$ be the closest linear function to $f$. Consider a set $C=X-x+y$ ($x\in X,y\in Y$ will be specified below). By definition $f(X)=3$ and $f(C)=0$. We argue that $f$ is no closer than $3/2$-close to $\mu$.
	
	Assume for contradiction that $f$ is $(3/2-\delta)$-close to $\mu$ for $\delta>0$. We show there is large enough $k$ for which this leads to contradiction ($k$ is the number of items in $X$, which equals the number of items in $Y$).
	
	First notice that $\mu(X) \ge 3/2+\delta$ (otherwise $f$ and $\mu$ are more than $3/2-\delta$ apart on $X$). If it were to hold that $\mu(x)\le\mu(y)$, we'd get that $\mu(C)=\mu(X)-\mu(x)+\mu(y)\ge 3/2+\delta$ (i.e., $f$ and $\mu$ are more than $3/2-\delta$ apart on $C$). So $\mu(x)>\mu(y)$ for every $x,y$.
	
	We know that $\mu(X)$ can't be too big (since it can't exceed $f(X)=3$ by too much). For the same reason and since $f(Y)=-3$ by definition, $\mu(Y)$ is negative but can't be too small. By letting $k$ grow large, it becomes apparent that $\mu(x)$ cannot be bounded away from 0 from above, similarly $\mu(y)$ cannot be bounded away from 0 from below. Using $\mu(x)>\mu(y)$ we conclude that $\mu(x)\ge0$ for every $x$ and $\mu(y)\le 0$ for every $y$, and that for most items $\mu(x)\to0$ and $\mu(y)\to 0$.
	
	Now choose $x$ and $y$ that minimize $|\mu(x)|+|\mu(y)|$. For large enough $k$ we get that $\mu(x)-\mu(y)<\delta$, so $\mu(C)=\mu(X)-\mu(x)+\mu(y)> 3/2 + \delta - \delta>f(C)+3/2-\delta$, contradiction. The bounds asserted in the theorem follow.
\end{proof}

\section{Supplement to Section \lowercase{\ref{sec:upper}} (Upper Bounds)}
\label{appx:upper}
\subsection{Expanders}
\label{appx:expanders}

In this appendix we give a proof sketch for Lemma \ref{lem:expander}, establishing the existence of expanders with a range of parameters. We begin with a claim following from Stirling's approximation. The proof of this claim appears for completeness in~\cite{FFT16}.

\begin{claim}[Stirling]
	\label{cla:stirling}
	For every $c,d\in \mathbb{R}$ where $c>d>0$, for every sufficiently large integer $m$ such that $cm$ and $dm$ are integers,
	$$
	{cm \choose dm} =
	\left(\frac{c^c}{d^d(c-d)^{c-d}}\right)^m	\Theta\left(\frac{1}{\sqrt{m}}\right).
	$$
\end{claim}


\begin{proof}[Proof of Lemma \ref{lem:expander} (Sketch)]
	We give the proof for $(\frac{1}{4},5,\frac{1}{2})$-expanders. Existence for the other families in Lemma \ref{lem:expander} follows from the same argument.
	Let $k=2m$. Suppose that there are $4m$ left-hand side vertices and consequently $2m$ right-hand side vertices (because $\theta = \frac{1}{2}$). The edges are determined by making $r=5$ copies of each left-hand side vertex, taking a random permutation over all $2kr=20m$ copies, and connecting to the right-hand side in a round-robin fashion. The probability that there is a set of $2k\alpha=m$ left-hand side vertices with fewer than $m$ neighbors is at most:
 	\begin{eqnarray*}
	{2k\choose 2k\alpha}{2k\theta \choose 2k\alpha}
	\frac{{2kr\frac{\alpha}{\theta} \choose 2kr\alpha}}{{2kr \choose 2kr\alpha}} = 
	{4m \choose m}{2m \choose m}\frac{{10m \choose 5m}}{{20m \choose 5m}}\\
	= 
	\left(\frac{4^4}{3^3 \cdot 1^1}\cdot\frac{2^2}{1^1 \cdot 1^1}\cdot\frac{10^{10}}{5^5 \cdot 5^5}\cdot\frac{5^5\cdot 15^{15}}{20^{20}}\right)^m \Theta\left(\frac{1}{m}\right)\\
	=
	\left(\frac{27}{32} \right)^{4m}\Theta\left(\frac{1}{m}\right)
	< 1,
	\end{eqnarray*}
	where the second equality follows from Claim \ref{cla:stirling}, and the inequality holds by taking sufficiently large $m$ to overcome the constants in $\Theta(1/m)$.
	Hence such expanders exist for every integer multiple $k$ of sufficiently large $k'$.
\end{proof}

\subsection{Strengthening the Upper Bound on $\Kalweak$}
\label{appx:UB-weak}

In this appendix we use the notation of Section \ref{sub:UB-weak}, and in addition denote the average deficit and surplus of the target sets of $\PS^*$ and $\NS^*$ by $d'$ and $s'$, respectively. We first state a stronger version of Lemma \ref{lem:UB-main-lemma}:

\begin{lemma}[Strong Version of Lemma \ref{lem:UB-main-lemma}]
	\label{lem:UB-main-lemma-full}
	Let $k\in \mathbb{N}_{\ge 0}$ and $\alpha,r,\theta\in \Replus$ be such that there exists an $(\alpha,r,\theta)$-expander $G_k$.
	Let $\mathcal{G}$ and $\mathcal{G}'$ (possibly $\mathcal{G}=\mathcal{G}'$) be collections of $2k$ sets each, both with $\alpha$-frequent items, such that the average deficit of $\mathcal{G}$ is $\le d$, the average surplus of $\mathcal{G}'$ is $\le s$, and their target sets have average deficit $\ge d'$ and average surplus $\ge s'$, respectively. Then
	\begin{equation*}
	M\le \frac{\frac{1}{2}(d+s- \theta (d'+ s')) + 2\epsilon(r - 1)}{1-\theta} + \epsilon.
	\end{equation*}
\end{lemma}

\begin{proof}
	We first apply Lemma~\ref{lem:KR} to partition and disjointly recombine the sets of collection $\mathcal{G}$ via the expander $G_k=G_k(V,W;E)$. We use the following notation: For a vertex $v\in V$ (resp., $w\in W$) let $S_v$ (resp., $S_w$) be the source (resp., target) set corresponding to $v$ (resp., $w$). Denote the neighboring vertices of a vertex $v$ by $N(v)$ and its degree by $\deg(v)$. Let $S_{v,w}$ be the intermediate subset that labels (corresponds to) edge $(v,w)\in E$.
	
	By Lemma \ref{lem:KR}, for every $v\in V$, the intermediate subsets labeling the edges adjacent to $v$ are disjoint, and the same holds for every $w\in W$. We can thus apply Observation \ref{obs:iterative} to get
	\begin{eqnarray*}
		&\sum_{w \in N(v)} f(S_{v,w}) \ge f(S_v) + (\deg(v) - 1)(f(\emptyset) - \epsilon)&\text{~~~}\forall v\in V.\\
		&\sum_{v \in N(w)} f(S_{v,w}) \le f(S_w) + (\deg(w) - 1)(f(\emptyset) + \epsilon)&\text{~~~}\forall w\in W.
	\end{eqnarray*}
	Denote the maximum absolute value of $f$ by $M$. Since collection $\mathcal{G}$ has average deficit $\le d$, by summing over vertices $v\in V$ (where $|V|=2k$) we get $\sum_{v\in V}f(S_v)\ge 2k(M-d)$. Since the target sets of $\mathcal{G}$ have average deficit $\ge d'$, by summing over vertices $w\in W$ (where $|W|=2\theta k$) we get $\sum_{w\in W}f(S_w)\le 2\theta k(M-d')$. Clearly in the bipartite graph $G_k$, $\sum_{v\in V}\deg(v)=\sum_{w\in W}\deg(w)$, and by the parameters of $G_k$ both are equal to $2kr$. Therefore, summing over $v \in V$ and $w\in W$ we get
	\begin{eqnarray*}
	2kM - 2kd + (2kr - 2k)(f(\emptyset) - \epsilon)
	\le \sum_{(v,w)\in E} f(S_{v,w})\\
	\le 2\theta kM - 2\theta kd' + (2kr - 2\theta k)(f(\emptyset) + \epsilon).
	\end{eqnarray*}
	Dividing the resulting inequality by $2k$ and rearranging gives
	\begin{eqnarray}
	(1 - \theta)M
	&\le&
	d - \theta d' + (r - 1)(\epsilon - f(\emptyset)) + (r - \theta)(\epsilon + f(\emptyset))\nonumber\\
	&=& d - \theta d' + 2\epsilon(r - 1) + (1-\theta)(\epsilon+f(\emptyset)).\label{eq:positive}
	\end{eqnarray}
	
	Similarly, using that the average deficit of collection $\mathcal{G}'$ is $\le s$ and the average deficit of its target sets is $\ge s'$, $\sum_{v\in V}f(S_v)\le 2k(-M+s)$, and $\sum_{w\in W}f(S_w)\ge -2\theta k(M-s')$. Therefore
	\begin{eqnarray*}
	- 2kM + 2ks + (2kr - 2k)(f(\emptyset) + \epsilon)
	\ge \sum_{(v,w)\in E} f(S_{v,w})\\
	\ge -2\theta kM + 2\theta k s' + (2kr - 2\theta k)(f(\emptyset) - \epsilon).
	\end{eqnarray*}
	Dividing the resulting inequality by $2k$ and rearranging gives
	\begin{eqnarray}
	(1 - \theta)M
	&\le&
	s - \theta s' + (r - 1)(\epsilon + f(\emptyset)) + (r - \theta)(\epsilon - f(\emptyset))\nonumber\\
	&=& s - \theta s' + 2\epsilon(r - 1) + (1-\theta)(\epsilon-f(\emptyset)).
	\label{eq:negative}
	\end{eqnarray}
	
	Rearranging Inequalities \eqref{eq:positive} and \eqref{eq:negative} as well as averaging the resulting inequalities implies the theorem.
\end{proof}

We use this lemma to show the following upper bound.

\begin{lemma}[Upper Bound on $\Kalweak$]
\label{lem:K}
Suppose that for fixed $r,r',\theta,\theta'\in \Replus$ there exist $(\frac{1}{2},r,\theta)$-expanders and $(1-\frac{1}{2\theta},r',\theta')$-expanders. Then the weak Kalton constant satisfies:
$$\Kalweak \le \min\left\{
\frac{ 2r - \frac{1}{2} - \theta - \theta\frac{d'+s'}{2} } {1-\theta},
\frac{ 2r' - \theta' +\frac{d'+s'}{2} } {1-\theta'}
\right\}.$$
\end{lemma}

\begin{proof}
Observe that there exists $k$ such that $f$ has collections $\PS^*$ and $\NS^*$ of $2k$ sets as described above, and expanders $G_k$ and $G'_k$ with the above parameters, respectively, exist. The first upper bound is found in Lemma  \ref{lem:K-first}. Now consider the collection $\mathcal{G}$ that is the \emph{complement of the target sets} of $\PS^*$. Since $\PS^*$ has $\frac{1}{2}$-frequent items, then the target sets have $\frac{1}{2\theta}$-frequent items (Lemma \ref{lem:KR}), and $\mathcal{G}$ has $(1-\frac{1}{2\theta})$-frequent items (Observation \ref{obs:freq-in-complement}). By Corollary \ref{cor:complement-deficit} with $\epsilon=1$, the average surplus of $\mathcal{G}$ is $\le d'+1+\delta$, and clearly the average surplus of its target sets is $\ge 0$. Similarly, let $\mathcal{G}'$ be the complement of the target sets of $\NS^*$, which has $(1-\frac{1}{2\theta})$-frequent items and average deficit $\le s'+1-\delta$. Clearly the average deficit of the target sets of $\mathcal{G}'$ is $\ge 0$. Applying Lemma \ref{lem:UB-main-lemma-full} thus gives an upper bound of
\begin{equation}
\Kalweak \le \frac{\frac{1}{2}(d'+s'+2) +2r' -2}{1-\theta'} + 1 = \frac{2r'-1+\frac{d'+s'}{2}}{1-\theta'} + 1.\label{eq:tweak}
\end{equation}

Taking the minimum of the bound in Lemma  \ref{lem:K-first} and the bound in \eqref{eq:tweak} completes the proof.
\end{proof}

We can now prove the upper bound in Theorem \ref{thm:UpperLowerWeak}, by which the weak Kalton constant satisfies $\Kalweak \le 23.811$.

\begin{proof}[Proof of First Part of Theorem \ref{thm:UpperLowerWeak}]
By using the expanders with $\alpha=\frac{1}{2}$ from Lemma \ref{lem:expander} we can substitute $r=5$ and $\theta=\frac{5}{7}$. Since $1-\frac{1}{2\theta}=\frac{3}{10}$, we can complement these parameters by using the expanders with $\alpha'=\frac{3}{10}$ from Lemma \ref{lem:expander} for which $r'=4$ and $\theta'=\frac{4}{7}$.
This implies that $\Kalweak \le \min\{30.75 -1.25(d'+s'), \frac{52}{3}+\frac{7}{6}(d'+s')\}=23.811$.
\end{proof}

\section{Lower Bound on $\Kalweak$}
\label{sub:LB-weak}

We prove the lower bound stated in Theorem \ref{thm:UpperLowerWeak} on $\Kalweak$.

\begin{proof}[Proof of Second Part of Theorem \ref{thm:UpperLowerWeak}]
	We show there exists a weakly 2-modular set function with $n=20$ that is tightly 3-linear. Thus, $\Kalweak \ge \frac{3}{2}$.
	Consider the following $(k,M)$-symmetric function $f$ with $k=3$ (hence with $n = {6 \choose 3} = 20$ items) and $M=3$. The positive support sets  ($\PS$) form a canonical generating collection. The values of sets under $f$ is determined by the first applicable rule (and hence $f$ is well defined):
	
	\begin{enumerate}
		
		\item  Each positive support set $S \in \PS$ has value $f(S) = M = 3$.
		
		\item If there is some set $P \in \PS$ such that $S \subset P$ and for every set $N \in \NS$ it holds that $S \not\subset N$, then $f(S) = 1$. Likewise (enforcing the dual symmetry property of $(k,M)$-symmetric functions), if there is some set $P \in \PS$ such that $P \subset S$ and for every set $N \in \NS$ it holds that $N \not\subset S$, then $f(S) = 1$.
		
		\item Enforcinging the antisymmetry property of $(k,M)$-symmetric functions, we impose $-f(S) = f(\bar{S})$ and derive from this sets with negative value.
		
		\item All other sets have value~0.
		
	\end{enumerate}
	
	One can verify (proof omitted) that the function $f$ defined above is indeed  $(k,M)$-symmetric according to Definition~\ref{def:(k,M)symmetric}. The reason why we choose $k = 3$ (and not smaller) is because we shall use the following claim.
	
	\begin{claim}
		\label{claim:weak}
		For every two sets $P_1, P_2 \in \PS$ and two sets $N_3, N_4 \in \NS$ all the following hold:
		
		\begin{enumerate}
			
			\item $P_1 \cap P_2 \not\subset N_3$.
			
			\item $N_3 \cap N_4 \not\subset P_1$.
			
			\item $P_1 \not\subset N_3 \cup N_4$. 
			
			\item $N_3 \not\subset P_1 \cup P_2$. 
			
		\end{enumerate}
	\end{claim}
	
	\begin{proof}
		We prove only item~1, as the proofs of the remaining items are similar. Suppose without loss of generality that $P_i$ (for $i \in \{1,2\}$) contains those items whose vector representation in $\{\pm 1\}^{2k}$ has bit~$i$ set to~1, and that $N_3$ contains those items whose vector representation in $\{\pm 1\}^{2k}$ has bit~3 set to $-1$. (If in $N_3$ it was bit~1 that was set to $-1$, the proof would be immediate.) Then because $k \ge 3$ the set $P_1 \cap P_2$ contains a vector for which all three bits~1, 2 and~3 are set to~1, and hence is not in~$N_3$.
	\end{proof}
	
	Proposition~\ref{pro:distM} implies that $f$ is tightly 3-linear. In remains to show that $f$ is weakly 2-modular. Consider two disjoint nonempty sets $S$ and $T$. As $S \cap T = \emptyset$, we have that $f(S \cap T) = 0$. Hence to prove weak 2-modularity, one needs to show that $|f(S) + f(T) - f(S\cup T)| \le 2$. Similar to the proof of Lemma~\ref{lem:shortcaseanalysis} it can be shown that one can assume without loss of generality that $f(S) \ge 0$ and $f(S) \ge |f(T)|$. We proceed by a case analysis.
	
	\begin{enumerate}
		
		\item $f(S) = 3$ (namely, $S \in \PS$).  In this case, because $T$ is disjoint from $S$, we have that $T \subset \bar{S}$ where $\bar{S} \in \NS$. The definition of $f$ then implies that $f(T) \le 0$.  We consider three cases.
		
		\begin{enumerate}
			
			\item $f(T) = 0$. Given that $T \subset \bar{S}$ and $\bar{S} \in \NS$, for $f(T) = 0$ to hold there must be some set $P \in \PS$ such that $T \subset P$. We claim that $f(S \cup T) = 1$ (which establishes weak 2-modularity, because $f(S) + f(T) = 3$). Suppose for the sake of contradiction that $f(S \cup T) \not= 1$. As $S \subset S \cup T$ and $S \in \PS$, this can happen only if there is a set $N \in \NS$ such that $N \subset S \cup T$. But then $N \cap \bar{S} \subset T \subset P$. Hence we found two sets in $\NS$ whose intersection lies in a set in $\PS$, contradicting item~2 of Claim~\ref{claim:weak}.
			
			\item $f(T) = -1$. In this case $f(S) + f(T) = 2$. We also have that $f(S \cup T) \ge 0$ because $S \subset S \cup T$ and $S \in \PS$. Hence regardless of the value of $f(S \cup T)$, weak 2-modularity holds.
			
			\item $f(T) = -3$. In this case $T = \bar{S}$ and $S \cup T = U$, implying that $f(S \cup T) = 0 = f(S) + f(T)$.
			
		\end{enumerate}
		
		\item $f(S) = 1$. Observe that the only way by which it may happen that $f(S \cup T) < 0$ is if there is some $P \in \PS$ such that $S \subset P$ (leading to $f(S) = 1$, and there is some $N \in \NS$ such that $N \subset S \cup T$, and $f(S \cup T) = -1$. (If $S \cup T \subset N$ then also $S \subset N$ and it cannot be that $f(S) = 1$.)
		
		\begin{enumerate}
			
			\item $f(T) = 1$.  We claim that $f(S \cup T) \ge 0$ (which implies that weak 2-modularity holds). Assume for the sake of contradiction that $f(S \cup T) = -1$. As noted above this implies that there is $P \in \PS$ with $S \subset P$, another $P' \in \PS$ with $T \subset P'$ (because we can interchange $S$ and $T$),  and $N \in NS$ such that $N \subset S \cup T$. But then $N \subset P \cup P'$, contradicting item~4 of Claim~\ref{claim:weak}.
			
			\item $f(T) = 0$. In this case $f(S)+ f(T) = 1$, and regardless of the value of $f(S \cup T)$ (which lies in the range $[-1,3]$), weak 2-modularity holds. 
			
			\item $f(T) = -1$. We need to show that $|f(S \cup T)| \not= 3$. This follows from the fact that $S$ in not contained in any set in $\NS$, and $T$ is not contained in any set in $\PS$.
			
		\end{enumerate}
		
		\item $f(S) = f(T) = 0$. We need to show that $|f(S \cup T)| \not= 3$. We show that $S \cup T \not\in \PS$ (and the proof that $S \cup T \not\in \NS$ is similar). Suppose for the sake of contradiction that $S \cup T = P$ with $P \in \PS$. Then $S \subset P$, and the fact that $f(S) = 0$ implies that there is a set $N_1 \in \NS$ such that $S \subset N_1$. Likewise, $T \subset N_2$ for some $N_2 \in \NS$. Hence $P \subset N_1 \cup N_2$, contradicting item~3 of Claim~\ref{claim:weak}.
	\end{enumerate}
	
	Hence we established that $f$ is weakly 2-modular, completing the proof.
\end{proof}

\section{Tightness of Learning Algorithm}
\label{appx:learning}

The following theorem establishes tightness of the learning algorithm.
Results similar to Theorem~\ref{thm:negative} appear in the literature -- see for example~\cite{SV}. See also \cite{CDDK} (Corollary 23).

\begin{theorem}
	\label{thm:negative}
	For $\Delta \le \sqrt{\log n / n}$ and $\delta=o(\sqrt{n / \log n})$, no $\delta$-learning algorithm exists for $\Delta$-linear set functions $f$, even if $f$ is monotone, and even if one allows for unbounded computation time (but only polynomially-many value queries). 
\end{theorem}

\begin{proof}[Proof (Sketch)] Consider a random linear function $g$ which contains a random set $T$ of cardinality $\frac{n}{2}$, in which each item has value $q = \Delta / \sqrt{n \log n}$, and the other items have value~0. Then $g(T) = (\Delta\sqrt{n}) / (2\sqrt{\log n})$ (which is at most~1), whereas $g(\bar{T}) = 0$.  Consider now a function $f$ constructed as follows. Call a set $S$ balanced if $||S \cap T| - \frac{|S|}{2}| \le \sqrt{n\log n}$, positively unbalanced if $|S \cap T| - \frac{|S|}{2} \ge \sqrt{n\log n}$, and negatively unbalanced if $\frac{|S|}{2} - |S \cap T| \ge \sqrt{n\log n}$. Then $f(S) = \frac{|S|}{2}q$ for balanced sets, $f(S) = g(S) - \sqrt{n\log n}$ for positively unbalanced sets, and $f(S) = g(S) + \sqrt{n\log n}$ for negatively unbalanced sets. 
	Observe that $g$ approximates $f$ because $q\sqrt{n\log n} = \Delta$.
	If there are only polynomially many queries then w.h.p., for every query $S$, the underlying set is balanced.
	Hence these queries are not informative in exposing $T$, and  $\delta \ge \frac{1}{2}(f(T) - f(\bar{T})) = \frac{\Delta\sqrt{n}}{4\sqrt{\log n}} - \Delta$.
\end{proof}

\section*{Acknowledgements}

We thank Assaf Naor for helpful discussions and for directing us to the paper of \cite{KR}. We thank Moni Naor for pointing us to the result of \cite{DY08}.

\bibliographystyle{abbrvnat}	
\bibliography{abb,modularity-bib}

\begin{thebibliography}{26}
\providecommand{\natexlab}[1]{#1}
\providecommand{\url}[1]{\texttt{#1}}
\expandafter\ifx\csname urlstyle\endcsname\relax
  \providecommand{\doi}[1]{doi: #1}\else
  \providecommand{\doi}{doi: \begingroup \urlstyle{rm}\Url}\fi

\bibitem[Balcan and Harvey(2011)]{BH11}
M.-F. Balcan and N.~J.~A. Harvey.
\newblock Learning submodular functions.
\newblock In \emph{Proceedings of the 43rd Annual ACM Symposium on Theory of
  Computing}, pages 793--802, 2011.

\bibitem[Balkanski et~al.(2017)Balkanski, Rubinstein, and Singer]{BRS16}
E.~Balkanski, A.~Rubinstein, and Y.~Singer.
\newblock The limitations of optimization from samples.
\newblock In \emph{Proceedings of the 48th Annual ACM Symposium on Theory of
  Computing}, pages 1016--1027, 2017.

\bibitem[Belloni et~al.(2015)Belloni, Liang, Narayanan, and Rakhlin]{BLNR15}
A.~Belloni, T.~Liang, H.~Narayanan, and A.~Rakhlin.
\newblock Escaping the local minima via simulated annealing: Optimization of
  approximately convex functions.
\newblock In \emph{Proceedings of the 28th Conference on Learning Theory},
  pages 240--265, 2015.

\bibitem[Bertsimas and Thiele(2014)]{BT14}
D.~Bertsimas and A.~Thiele.
\newblock \emph{Robust and data-driven optimization: Modern decision making
  under uncertainty}, chapter~5, pages 95--122.
\newblock INFORMS PubsOnline, 2014.
\newblock TutORials in Operations Research.

\bibitem[Blum et~al.(1993)Blum, Luby, and Rubinfeld]{BLR93}
M.~Blum, M.~Luby, and R.~Rubinfeld.
\newblock Self-testing/correcting with applications to numerical problems.
\newblock \emph{J. Comput. Syst. Sci.}, 47\penalty0 (3):\penalty0 549--595,
  1993.

\bibitem[Bondarenko et~al.(2013)Bondarenko, Prymak, and Radchenko]{BPR}
A.~V. Bondarenko, A.~Prymak, and D.~Radchenko.
\newblock On concentrators and related approximation constants.
\newblock \emph{J. Math. Anal. Appl.}, 402\penalty0 (1):\penalty0 234–--241,
  2013.

\bibitem[Chierichetti et~al.(2015)Chierichetti, Das, Dasgupta, and Kumar]{CDDK}
F.~Chierichetti, A.~Das, A.~Dasgupta, and R.~Kumar.
\newblock Approximate modularity.
\newblock In \emph{Proceedings of the 56th Symposium on Foundations of Computer
  Science}, pages 1143--1162, 2015.

\bibitem[Dwork and Yekhanin(2008)]{DY08}
C.~Dwork and S.~Yekhanin.
\newblock New efficient attacks on statistical disclosure control mechanisms.
\newblock In \emph{Proceedings of the 28th Annual International Cryptology
  Conference}, pages 469--480, 2008.

\bibitem[Feige(2009)]{Fei09}
U.~Feige.
\newblock On maximizing welfare when utility functions are subadditive.
\newblock \emph{SIAM J. Comput.}, 39\penalty0 (1):\penalty0 122--142, 2009.

\bibitem[Feige and Izsak(2013)]{FeigeIzsak}
U.~Feige and R.~Izsak.
\newblock Welfare maximization and the supermodular degree.
\newblock In \emph{Proceedings of the 4th Innovations in Theoretical Computer
  Science}, pages 247--256, 2013.

\bibitem[Feige et~al.(2016)Feige, Feldman, and Talgam-Cohen]{FFT16}
U.~Feige, M.~Feldman, and I.~Talgam-Cohen.
\newblock Approximate modularity revisited.
\newblock Supplemented version, available at
  {\url{https://arxiv.org/abs/1612.02034}}, 2016.

\bibitem[Goemans et~al.(2009)Goemans, Harvey, Iwata, and Mirrokni]{GHIM09}
M.~X. Goemans, N.~J.~A. Harvey, S.~Iwata, and V.~S. Mirrokni.
\newblock Approximating submodular functions everywhere.
\newblock In \emph{Proceedings of the 20th Annual ACM-SIAM Symposium on
  Discrete Algorithms}, pages 535--544, 2009.

\bibitem[Hassidim and Singer(2017)]{HS16}
A.~Hassidim and Y.~Singer.
\newblock Submodular optimization under noise.
\newblock In \emph{Proceedings of the 34th International Conference on Machine
  Learning}, pages 1069--1122, 2017.

\bibitem[Hyers(1941)]{Hye41}
D.~H. Hyers.
\newblock On the stability of the linear functional equation.
\newblock \emph{PNAS}, 27:\penalty0 222--224, 1941.

\bibitem[Iwata et~al.(2001)Iwata, Fleischer, and Fujishige]{IwataFF01}
S.~Iwata, L.~Fleischer, and S.~Fujishige.
\newblock A combinatorial strongly polynomial algorithm for minimizing
  submodular functions.
\newblock \emph{J. {ACM}}, 48\penalty0 (4):\penalty0 761--777, 2001.
\newblock \doi{10.1145/502090.502096}.
\newblock URL \url{http://doi.acm.org/10.1145/502090.502096}.

\bibitem[Jung(2011)]{Jung11}
S.~M. Jung.
\newblock \emph{{H}yers-{U}lam-{R}assias stability of functional equations in
  nonlinear analysis}.
\newblock Springer, 2011.

\bibitem[Kalton and Roberts(1983)]{KR}
N.~Kalton and J.~W. Roberts.
\newblock Uniformly exhaustive submeasures and nearly additive set functions.
\newblock \emph{Transactions of the American Mathematical Society},
  278\penalty0 (2):\penalty0 803--816, 1983.

\bibitem[Krause and Cevher(2010)]{KC10}
A.~Krause and V.~Cevher.
\newblock Submodular dictionary selection for sparse representation.
\newblock In \emph{Proceedings of the 27th International Conference on Machine
  Learning}, pages 567--574, 2010.

\bibitem[Krause et~al.(2008)Krause, Singh, and Guestrin]{KSG08}
A.~Krause, A.~P. Singh, and C.~Guestrin.
\newblock Near-optimal sensor placements in gaussian processes: Theory,
  efficient algorithms and empirical studies.
\newblock \emph{Journal of Machine Learning Research}, 9:\penalty0 235--–284,
  2008.

\bibitem[Lehmann et~al.(2006)Lehmann, Lehmann, and Nisan]{LLN06}
B.~Lehmann, D.~Lehmann, and N.~Nisan.
\newblock Combinatorial auctions with decreasing marginal utilities.
\newblock \emph{Games and Economic Behavior}, 55:\penalty0 270--296, 2006.

\bibitem[{Paes Leme}(2017)]{Pae14}
R.~{Paes Leme}.
\newblock Gross substitutability: An algorithmic survey.
\newblock To appear in Games and Economic Behavior, 2017.

\bibitem[Pawlik(1987)]{Pawlik}
B.~Pawlik.
\newblock Approximately additive set functions.
\newblock \emph{Colloquium Mathematicae}, 54\penalty0 (1):\penalty0 163--164,
  1987.

\bibitem[Pippenger(1977)]{pippenger}
N.~Pippenger.
\newblock Superconcentrators.
\newblock \emph{SIAM J. Comput.}, 6\penalty0 (2):\penalty0 298--304, 1977.

\bibitem[Roughgarden et~al.(2017)Roughgarden, Talgam-Cohen, and
  Vondr\'ak]{TimInbalJan16}
T.~Roughgarden, I.~Talgam-Cohen, and J.~Vondr\'ak.
\newblock When are welfare guarantees robust?
\newblock In \emph{Proceedings of the 20th International Workshop on
  Approximation Algorithms for Combinatorial Optimization Problems}, pages
  22:1--22:23, 2017.

\bibitem[Schrijver(2000)]{Sch00}
A.~Schrijver.
\newblock A combinatorial algorithm minimizing submodular functions in strongly
  polynomial time.
\newblock \emph{Journal of Combinatorial Theory}, 80\penalty0 (2):\penalty0
  346--355, 2000.

\bibitem[Singer and Vondr\'ak(2015)]{SV}
Y.~Singer and J.~Vondr\'ak.
\newblock Information-theoretic lower bounds for convex optimization with
  erroneous oracles.
\newblock In \emph{Proceedings of the 28th Annual Conference on Neural
  Information Processing Systems}, pages 3204--3212, 2015.

\end{thebibliography}

\end{document}